%% file: document.tex
\documentclass[a4paper, 11pt]{scrartcl}

\input{./config/packages}
\input{./config/commands}

\title{\normalfont Quantum Risk Analysis: Beyond (Conditional) Value-at-Risk}

\author[a]{Christian Laudagé\thanks{\textit{E-mail:} christian.laudage@rptu.de (corresponding author).}}
\author[b]{Ivica Turkalj\thanks{\textit{E-mail:} ivica.turkalj@itwm.fraunhofer.de.}}

\affil[a]{Department of Mathematics, RPTU Kaiserslautern-Landau, Gottlieb-Daimler-Straße 47, 67663 Kaiserslautern, Germany}
\affil[b]{Department of Financial Mathematics, Fraunhofer Institute for Industrial Mathematics ITWM, Fraunhofer-Platz 1, 67663 Kaiserslautern, Germany}

\date{\today}

\begin{document}

	\pagenumbering{arabic}
	\numberwithin{equation}{section}\numberwithin{equation}{section}
	
	\maketitle
	
	\begin{abstract} 
		
		Risk measures are important key figures to measure the adequacy of the reserves of a company. The most common risk measures in practice are Value-at-Risk (VaR) and Conditional Value-at-Risk (CVaR). Recently, quantum-based algorithms are introduced to calculate them. These procedures are based on the so-called quantum amplitude estimation algorithm which lead to a quadratic speed up compared to classical Monte-Carlo based methods.
			
		Based on these ideas, we construct quantum-based algorithms to calculate alternatives for VaR and CVaR, namely the Expectile Value-at-Risk (EVaR) and the Range Value-at-Risk (RVaR). We construct quantum algorithms to calculate them. These algorithms are  based on quantum amplitude estimation. In two case studies, we compare their performance with the quantum-based algorithms for VaR and CVaR. We find that all of the algorithms perform sufficiently well on a quantum simulator. Further, the calculations of EVaR and VaR are robust against noise on a real quantum device. This is not the case for CVaR and RVaR. 
		
		\medskip
		
		\item[\hskip\labelsep\scshape Keywords:] Risk measures, expectiles, Range Value-at-Risk, quantum computing, quantum amplitude estimation
	
	\end{abstract}
	
	\section{Introduction} 
	
	Risk measures are important key figures used by financial institutions to assess their financial positions. The most common risk measures used in practice are the Value-at-Risk and the Conditional Value-at-Risk. In most of the cases, they are calculated by time-consuming Monte-Carlo based methods. In contrast,~\textcite{woerner_quantum_2019} use the so-called quantum amplitude estimation to calculate the Value-at-Risk and the Conditional Value-at-Risk. Compared to classical Monte-Carlo methods, this approach promises a quadratic reduction of the computation time. 
	
	There are reasonable alternatives for the Value-at-Risk and Conditional Value-at-Risk and to the best of our knowledge, they have not been calculated on a quantum computer until now. We would like to fill this gap for the following two alternatives: First, the Expectile Value-at-Risk and second, the Range Value-at-Risk. The Expectile Value-at-Risk is based on the statistical concept of expectiles. 
	
	\textbf{State-of-the-art:} Expectiles are introduced originally in~\textcite{newey_asymmetric_1987}. An expectile can be seen as a generalized mean and the term ``expectile'' is a combination of the words expectation and quantile. Hence, it is a compromise between Value-at-Risk and Conditional Value-at-Risk. Further, it is a special case of a shortfall risk measure as defined in, e.g.,~\textcite{follmer_stochastic_2016}. 
	
	Expectiles has been analyzed in many real world situations, in which they are meaningful alternatives for Value-at-Risk and Conditional Value-at-Risk. \textcite{bellini_risk_2017} compare expectiles with classical risk measures for the S\&P500 Index. As they mentioned, their  ``Theoretical and numerical results indicate that expectiles are perfectly reasonable alternatives to VaR and ES'', see~\textcite[page 487]{bellini_risk_2017} as well as the references therein. \textcite{KUAN2009261} tested expectiles for stock market data, namely the S\&P500 and the Nasdaq index. But also for other fields of applications expectiles are relevant, e.g.,~\textcite{DEROSSI2009179} apply expectiles in a context with motorcycle data. Beside their direct usage as risk measures, expectiles can also be used to estimate Value-at-Risk and Expected Shortfall, see~\textcite{taylor_2008}, i.e.,~they are also useful, even if other risk measures are applied.  

	Furthermore, unlike Value-at-Risk and Conditional Value-at-Risk, expectiles admit some desirable theoretical properties. They are the only law-invariant coherent risk measures that are elicitable, see~\textcite{bellini_risk_2017} or~\textcite{ziegel_2016}.	Furthermore, expectiles are the only $M$-quantiles (which are all elicitable) that are coherent, see~\textcite{bellini_generalized_2014}.
	
	In contrast to the Value-at-Risk, the Conditional Value-at-Risk takes the amount of losses into account for calculating capital requirements. But the latter is not qualitatively robust, i.e.,~if there are small perturbations in the law of the underlying data, then the law of the Conditional Value-at-Risk estimator can change significantly. A qualitatively robust risk measure that takes the amount of losses into account is the Range Value-at-Risk, which admits the Value-at-Risk and the Conditional Value-at-Risk as special cases. It is introduced in~\textcite{cont_robustness_2010}. \textcite{Righi_Mueller_2023} test among other risk measures the Range Value-at-Risk for the S\&P500 index. In the context of risk sharing, the usage of Range Value-at-Risk is discussed in~\textcite{embrechts_quantile-based_2020}. As mentioned by~\textcite[page 25]{fissler_elicitability_2021}, ``Range value at risk (RVaR) is a natural interpolation between VaR and ES, constituting a tradeoff between the sensitivity of ES and the robustness of VaR, turning it into a practically relevant risk measure on its own''. Since the estimation of the RVaR is an issue on its own in practice,~\textcite{computation11020028} compare different estimators for RVaR.
	
	We suggest procedures for the calculation of expectiles and Range Value-at-Risk by using quantum amplitude estimation. For an introduction to quantum amplitude estimation we refer to~\textcite{kaye_introduction_2007}. The application of quantum computing to solve tasks in financial mathematics is quite new. For an overview we refer to~\textcite{egger_quantum_2020}. Option pricing with the help of amplitude estimation is performed in~\textcite{stamatopoulos_option_2020} and~\textcite{chakrabarti_threshold_2021}. Market risk of financial derivatives are determined via quantum computing in~\textcite{stamatopoulos_towards_2022}. Algorithms for calculating risk measures are developed in~\textcite{woerner_quantum_2019}. These algorithms are applied to calculate credit risk in~\textcite{egger_credit_2021}. \textcite{stamatopoulos2024quantumriskanalysisfinancial} develop new algorithms to calculate risk measures via quantum signal processing (QSP). An application of QSP reduces the required resources, which is demonstrated in the context of financial derivative pricing in~\textcite{Stamatopoulos_2024}.
	
	\textbf{Methodology:} An expectile can be formulated as the root of a specific function including an expectation. The idea of computing expectiles by a quantum algorithm stems from the tractability of this function on a quantum computer. Hence, we calculate this function by quantum amplitude estimation and perform a root search algorithm on a classical computer. In contrast, the Range Value-at-Risk is directly calculated as the result of a quantum amplitude estimation.
	
	\textbf{Contributions of this manuscript:}
	First, we transfer the ideas to calculate Value-at-Risk and Conditional Value-at-Risk in~\textcite{woerner_quantum_2019} to build up new algorithms for expectiles and Range Value-at-Risk.  Here, the construction of the operators used by the quantum amplitude estimation is new, especially in the case of the expectiles, compare with Equation~(\ref{eq:expectiles_equivalent_formulation}),~\prettyref{sec:bisectionExpectiles} and~\prettyref{sec:quantum_algorithm_expectiles}. Second, we calculate these risk measures, as well as, Value-at-Risk and Conditional Value-at-Risk for two case studies: The first one carries the situation of a credit risk portfolio from a financial institution. In the second one, we calibrate lognormal and gamma distributions towards a real world data set from the insurance business and calculate the mentioned risk measures as well as compare their performances.
	
	For the first case study, we use a canonical quantum amplitude estimation and we find that an increasing number of ancilla qubits leads to an accuracy towards the true value in case of the expectiles. In contrast, the accuracy in case of the Range Value-at-Risk is less accurate. The example is based on two assets. In~\prettyref{rem:scalingRealWorldProblem} we argue that for more assets we can use known numbers from literature to extrapolate the circuit depth. 

	In the second case study we find that the algorithms converge in an adequate manner depending on the number of qubits used for loading the distributions. Furthermore, the efficiency of noisy intermediate-scale quantum devices are evaluated by applying different variants of the quantum amplitude estimation on a noisy simulator. The expectile calculation is robust against the noise of a real quantum hardware and an iterative variant of the quantum amplitude estimation leads to the lowest estimation errors. In contrast, the results for the Range Value-at-Risk are significantly affected by the noise from a quantum device.
	
	\textbf{Structure of this manuscript:} In~\prettyref{sec:definition_risk_measures}, we introduce the risk measures of interest. In~\prettyref{sec:bisection}, we state the algorithms for expectiles and Range Value-at-Risk. In~\prettyref{sec:quantum_algorithm}, the implementation of the operators used in the quantum amplitude estimation is described. In~\prettyref{sec:numerical_study}, we compare the performance of the algorithms for Value-at-Risk, Conditional Value-at-Risk, expectiles and Range Value-at-Risk via two numerical case studies. A conclusion and outlook is given in~\prettyref{sec:conclusion}. For the sake of completeness, the mathematical background of the analyzed risk measures and properties of expectiles that are important for applying the root search algorithm are stated in Appendix~\ref{sec:mathematicalBackground}. Information on the construction of the circuits to load a probability density function are postponed to Appendix~\ref{sec:distributionLoading}. Additional figures are shifted to Appendix~\ref{sec:figures}.

	\section{Risk measures}\label{sec:definition_risk_measures}
	
	We start by explaining the idea behind a risk measure. In practice, risk measures are used as key figures to describe the risk of a future unknown financial position. The value of a risk measure can then be interpreted as a capital reserve.
	
	\textcite{woerner_quantum_2019} discussed the Value-at-Risk (VaR) and the Conditional Value-at-Risk (CVaR). The latter coincides with the so-called Expected Shortfall (ES) in case of a continuous random variable. We aim to discuss two more risk measures, namely the Expectile Value-at-Risk (EVaR) and the Range Value-at-Risk (RVaR).
	
	In the following, we explain how they are calculated for a continuous random variable $X$ and mention their main characteristics. Probabilities are denoted by $\symbolProbabilityMeasure(.)$ and expectations are denoted by $\expectation{.}$. For completeness, we state the mathematical definitions in Appendix~\ref{sec:mathematicalBackground}. \\

	\textbf{VaR:} The VaR at level $0<\lambda<1$ is the negative value of the $\lambda$-quantile of the random variable $X$. We denote this value by $\valueAtRisk{\lambda}{X}$. Hence, it holds that $$\symbolProbabilityMeasure\left(X\leq -\valueAtRisk{\lambda}{X}\right) = \symbolProbabilityMeasure\left(X+\valueAtRisk{\lambda}{X}\leq 0\right) = \lambda.$$ 
	
	So, by using the VaR as capital reserve, the value of the total position $X+\valueAtRisk{\lambda}{X}$ is only in $\lambda$-percent of the cases negative, i.e.,~in only $\lambda$-percent of the cases, a loss occurs.
	
	But, the VaR does not encourage the diversification of a portfolio in general, i.e.,~diversification does not lead to a decrease of the capital reserve in general. Further, the VaR does not take the amount of losses into account. To avoid these undesirable effects, the CVaR became quite popular in practice. \\
	
	\textbf{CVaR:} The CVaR at level $0<\lambda<1$ takes the amount of losses into account in the sense that it is the negative conditional expectation of the random variable $X$, conditioned on the event that the value of $X$ is lower than $-\valueAtRisk{\lambda}{X}$. We denote this value by $\conditionalValueAtRisk{\lambda}{X}$ and it is given by
	\begin{align*}
		\conditionalValueAtRisk{\lambda}{X} = \expectation{-X\,|\,X \leq -\valueAtRisk{\lambda}{X}}.
	\end{align*} 
	
	So, the CVaR gives us the negative average of the $\lambda$-percent worst scenarios of $X$. \\ 
		
	\textbf{EVaR:} The EVaR is a valid alternative to the commonly used VaR and CVaR. The EVaR is the negative value of a so-called expectile. The word expectile is a combination of the words expectation and quantile. As stated by~\textcite{bellini_risk_2017}, ``expectiles can be seen as an asymmetric generalization of the mean'' and we can interprete them as a compromise between VaR and CVaR. We use the following quote from~\textcite{bellini_risk_2017} to motivate that expectiles are good alternatives for VaR and CVaR:
	\begin{quote}
		``Theoretical and numerical results indicate that expectiles are perfectly reasonable alternatives to VaR and ES risk measures.''
	\end{quote}
	
	The expectile at level $0<\alpha<1$ of the random variable $X$ is denoted by $\expectile{\alpha}{X}$ and it is the solution of the following equation:
	\begin{align}\label{eq:expectiles}
		\level\expectation{\max\{X-\expectile{\level}{X},0\}}=(1-\level)\expectation{-\min\{X-\expectile{\level}{X},0\}}.
	\end{align}
	
	A good interpretation of this value is given in the case of $\expectation{|X-\expectile{\level}{X}|}\neq 0$. Then, we can rewrite the previous equation as
	\begin{align*}
		\alpha = \frac{\expectation{-\min\{X-\expectile{\level}{X},0\}}}{\expectation{|X-\expectile{\level}{X}|}}.
	\end{align*}
	Hence, the expectile is the value for which the expected loss of  $X-\expectile{\alpha}{X}$ only accounts for $\alpha$-percent of the total deviations of $X-\expectile{\alpha}{X}$ from zero. So, we can interprete $-\expectile{\alpha}{X}$ as a capital reserve and this is exactly the EVaR, which we denote by $\expectileVaR{\level}{X}= -\expectile{\level}{X}$.\\
	
	\textbf{RVaR:} An important property for the application of risk measures is qualitative robustness. It describes roughly that small changes in the law of the observed data points should not lead to drastic changes in the law of the risk estimator. Changes are measured by an appropriate metric. For a precise definition of qualitative robustness in the context of risk measures we refer to~\textcite{cont_robustness_2010, kratschmer_comparative_2014, koch-medina_qualitative_2022}.
	
	The VaR is qualitatively robust, while the CVaR is not. But the CVaR takes the magnitude of losses into account. As a compromise, one can use the RVaR, which admits desirable robustness properties in the sense of~\textcite{cont_robustness_2010} and takes the amount of losses partly into account.
	
	The RVaR is defined for two levels $0<\alpha<\beta<1$. We denote it by $\rangeValueAtRisk{\alpha,\beta}{X}$ and it is given by
	\begin{align*}
		\rangeValueAtRisk{\alpha,\beta}{X} = \expectation{-X\,|\,-\valueAtRisk{\alpha}{X}\leq X \leq  -\valueAtRisk{\beta}{X}}.
	\end{align*} 
	
	In contrast to the CVaR, we only consider a specific part of the tail of the distribution of $X$. This makes the RVaR robust against perturbations in the underlying data.

	\section{Algorithms}\label{sec:bisection}
	
	In this section, we sketch the algorithms that are used to calculate the EVaR (or rather expectiles) and the RVaR. Additionally, we point out at which parts in these algorithms we apply quantum computing. Details about the used quantum techniques are given later in~\prettyref{sec:quantum_algorithm}. 
	
	For the expectiles we perform a bisection search algorithm on a classical device and in every iteration step, the objective function is calculated on a quantum computer. This is motivated by the procedure for the VaR in~\textcite{woerner_quantum_2019}. 
	
	For the RVaR, we perform four operations on the quantum computer. The first two are used to obtain the VaR-values for the specified levels. The third gives us the probability to lie between these two VaR-values and the last one calculates the conditional expectation which coincides with the RVaR. This idea origins from the procedure for the CVaR in~\textcite{woerner_quantum_2019}. 
	
	\subsection{Algorithm: Expectiles}\label{sec:bisectionExpectiles}
	
	Our aim is to obtain an expectile via a root search algorithm. Note that the solution of Equation~\prettyref{eq:expectiles} is unique, see e.g.,~\textcite[Proposition 1]{bellini_generalized_2014}. Then, for the implementation on the quantum computer we use the following equivalent representation of~\prettyref{eq:expectiles} in the case of $\alpha\geq\frac{1}{2}$ with $\beta=\frac{2\level-1}{1-\level}$:
	\begin{align}\label{eq:expectiles_equivalent_formulation}
		\expectile{\level}{X} = \expectation{\max\{(1+\beta)X-\beta\expectile{\level}{X},X\}}.
	\end{align}
	
	This representation follows by using the fact that $X=\max\{X,0\}-\max\{-X,0\}$. 
 	Note, in the case of $\alpha <\frac{1}{2}$, we can also apply the previous representation by using the fact that $\expectile{\alpha}{X} = -\expectile{1-\alpha}{-X}$. 
	
	Motivated by the right-hand side of the previous equation, let us introduce the following function for $\alpha\geq \frac{1}{2}$: 
	\begin{align}\label{eq:auxiliary_function_expected_payoff}
		h_{X,\level}:\realLine\rightarrow\realLine,x\mapsto\expectation{\max\left\{\left(1+\beta\right)X-\beta x,X\right\}}.
	\end{align}
	
	For $x=\expectile{\level}{X}$ the value of $h_{X,\level}$ is equal to the right hand side in~\prettyref{eq:expectiles_equivalent_formulation}. Therefore, performing a bisection search algorithm until $h_{X,\level}\left(x\right)\approx x$ gives us an approximation of the expectile $\expectile{\level}{X}$. In the end, we determine $h_{X,\level}$ via a quantum algorithm.
	
	This procedure leads to Algorithm 1. In there, we marked in red the expressions that are calculated using a quantum computer, as described in~\prettyref{sec:quantum_algorithm_expectiles}. In this algorithm, $a$ and $b$ are the minimum and maximum values of the discretized version of the random variable (compare with~\prettyref{sec:quantum_algorithm}). $N$ is the maximum number of iterations for the bisection search algorithm. This algorithm should stop even earlier by reaching prespecified tolerance values $\epsilon,\delta>0$.
	
	\begin{algorithm}
		\textbf{Algorithm 1} Expectiles
		
		\rule[1ex]{\linewidth}{0.4pt}
		\begin{algorithmic}
			\State $x_1 \gets {\color{BrickRed}h_{X,\level}(a)} - a$
			\State $x_2 \gets {\color{BrickRed}h_{X,\level}(b)} - b$
			\For{$i = 1$ to $N$}
			\State $x \gets \frac{x_1 + x_2}{2}$
			\State $y \gets {\color{BrickRed}h_{X,\level}(x)} - x$
			\If{$\abs{y}<\epsilon$ or $\abs{\frac{x_2-x_1}{2}}<\delta$}
			\State $\expectile{\level}{X} \gets x$
			\State \textbf{break}
			\EndIf
			\If{y>0}
			\State $x_1 \gets x$
			\Else
			\State $x_2 \gets x$
			\EndIf
			\EndFor
		\end{algorithmic}
	\end{algorithm}
	
	In order that the bisection search algorithm is well-defined, we have to check the following two properties of the map $h_{X,\level}$: First, we have to prove that it is possible to choose the starting values $x_1$ and $x_2$ in Algorithm 1 such that they have different signs. Second, the map $x\mapsto h_{X,\level}(x)-x$ has to be continuous. For the sake of brevity, we shift these theoretical results to Appendix~\ref{sec:mathematicalBackground}.
	
	\subsection{Algorithm: Range Value-at-Risk}
	
	Recall that for a continuous random variable $X$ the RVaR at levels $\level$ and $\beta$ is the following conditional expectation:
	\begin{align}\label{eq:conditionalExpectationRVaR}
		\rangeValueAtRisk{\level,\beta}{X}=\expectation {-X\,|\,-\valueAtRisk{\alpha}{X}\leq X \leq -\valueAtRisk{\beta}{X}}.
	\end{align} 
	
	We can calculate this conditional expectation directly on the quantum computer. This is illustrated by Algorithm 2. Again, we marked in red the values that are calculated via quantum computing. For $x_1$, $x_2$ and $p$, we use the procedures described in~\textcite{woerner_quantum_2019}. 
	After the calculation of $x_1$, $x_2$ and $p$ we do not have to use an iterative procedure to obtain the RVaR.
	For the calculation of the RVaR see~\prettyref{sec:quantum_algorithm_rvar}. In there, it is also explained how the probability $p$ enters the calculation of the RVaR.
	\begin{algorithm}
		\textbf{Algorithm 2} Range Value-at-Risk
		
		\rule[1ex]{\linewidth}{0.4pt}
		\begin{algorithmic}
			\State $x_1 \gets {\color{BrickRed}\valueAtRisk{\alpha}{X}}$
			\State $x_2 \gets {\color{BrickRed}\valueAtRisk{\beta}{X}}$
			\State $p \gets {\color{BrickRed}\symbolProbabilityMeasure(-x_1\leq X\leq -x_2)}$
			\State $\rangeValueAtRisk{\alpha,\beta}{X} \gets {\color{BrickRed}\expectation {-X\,|\,-x_1\leq X\leq -x_2}}$ \Comment{Calculation requires knowledge of $p$}
		\end{algorithmic}
	\end{algorithm}

	\section{Implementation on a quantum computer}\label{sec:quantum_algorithm}
	
	In this section, we describe the implementation of the operators used by a quantum algorithm to estimate the expectations in~\prettyref{eq:auxiliary_function_expected_payoff} and~\prettyref{eq:conditionalExpectationRVaR}. The quantum algorithm that we use is the so-called amplitude estimation algorithm. In~\prettyref{sec:qae} we discuss important facts of different versions of amplitude estimation. Then in Sections~\ref{sec:quantum_algorithm_expectiles} and~\ref{sec:quantum_algorithm_rvar} we state details for the amplitude estimation for expectiles and RVaR.
	
	\subsection{Quantum amplitude estimation}\label{sec:qae}
	
	For a detailed introduction to quantum amplitude estimation we refer to~\textcite[Chapter 8]{kaye_introduction_2007}. We illustrate the canonical quantum circuit of quantum amplitude estimation (QAE) in~\prettyref{fig:amplitude_estimation_circuit}, compare with, e.g.,~\textcite{brassard_quantum_2000}. 
	
	\begin{figure}[!ht]
		\[
		\begin{array}{c}
			\Qcircuit @C=1.5em @R=0.75em {
				&\lstick{(m-1)\quad \ket{0}_{\ }}&\gate{H}& \qw & \qw & & & \ctrl{4} & \multigate{3}{\text{QFT}^{-1}} & \meter\\
				&  & \vdots &  & & & &  & & \vdots \\
				&  &  &  &  & & & & & \\
				&\lstick{(0)\quad \ket{0}_{\ }}  & \gate{H}  & \ctrl{1} & \qw & & & \qw &\ghost{\text{QFT}^{-1}} & \meter\\ 
				&\lstick{\ket{i}_{n}}&\multigate{2}{\mathcal{A}}&\multigate{2}{Q^0}& \qw & \cdots &  & \multigate{2}{Q^{m-1}} & \qw \\
				&\lstick{\ket{c}_{\ }}  &\ghost{\mathcal{A}}  & \ghost{Q^0} & \qw & & & \ghost{Q^{m-1}} & \qw\\
				&\lstick{\ket{0}_{\ }}  &\ghost{\mathcal{A}}  & \ghost{Q^0} &\qw & & & \ghost{Q^{m-1}} & \qw
			}
		\end{array}
		\]
		\caption{Quantum circuit of amplitude estimation. The distribution of $X$ is loaded into the $\ket{i}_{n}$ register, which is initially in state $\ket{0}_{n}$. Loading the distribution is part of the operator $\mathcal{A}$. The $\ket{c}$ register is an ancilla qubit,  initially in state $\ket{0}$, which is used for a comparator circuit included in $\mathcal{A}$. This comparator circuit is also based on $n$ additional ancilla qubits. For the sake of brevity, we omit them here. $H$ denotes the Hadamard gate and $\text{QFT}^{-1}$ the inverse of the quantum Fourier transform.}
		\label{fig:amplitude_estimation_circuit}
	\end{figure}
	
	For the sake of convenience, we explain the idea behind QAE. To do so,  assume that $X$ is a discrete random variable attaining almost surely the values $\{0,1,\dots,2^n-1\}$, where $n\in\naturalNumbers$. Further, $X$ attains the value $i\in\{0,1,\dots,2^n-1\}$ with probability $p_i$. Then, we can load the distribution of $X$ by using $n$ qubits (analogously to bits on a classical computer) on a quantum computer. This is done by applying the following operator $\mathcal{R}$ to the first $n$ qubits:
	\begin{align*}
		\mathcal{R}\ket{0}_n = \sum_{i = 0}^{2^n-1}\sqrt{p_i}\ket{i}_n.
	\end{align*} 
	
	For details on our used implementation to construct the operator $\mathcal{R}$, see Appendix~\ref{sec:distributionLoading}.
	
	Assuming a function $f:\{0,1,\dots,2^n-1\}\rightarrow[0,1]$, we can construct a new operator $\mathcal{C}$ on the first $n+1$ qubits such that
	\begin{align*}
		\mathcal{C}\ket{i}_n\ket{0} = \ket{i}_n\left(\sqrt{1-f(i)}\ket{0}+\sqrt{f(i)}\ket{1}\right).
	\end{align*}
	
	By applying the operator $\mathcal{R}$ to the first $n$ qubits and then the operator $\mathcal{C}$ to the first $n+1$ qubits, we obtain the following state
	\begin{align*}
		\sum_{i = 0}^{2^n-1}\sqrt{(1-f(i))p_i}\ket{i}_n\ket{0}+\sum_{i = 0}^{2^n-1}\sqrt{f(i)p_i}\ket{i}_n\ket{1}.
	\end{align*}
	
	This is also the form of the state that we obtain by performing the operator $\mathcal{A}$  in~\prettyref{fig:amplitude_estimation_circuit}. Now, QAE allows us to approximate the probability of the last qubit to be in state $\ket{1}$. This probability is given by
	\begin{align*}
		\sum_{i = 0}^{2^n-1}f(i)p_i = \expectation{f(i)}.
	\end{align*}
	
	Hence, by an appropriate choice of the map $f$, the QAE can be used to obtain estimators for the expectations in Equations~\prettyref{eq:auxiliary_function_expected_payoff} and~\prettyref{eq:conditionalExpectationRVaR}. 
	
	The concrete forms of the operator $\mathcal{A}$ for expectiles and RVaR are described in Sections~\ref{sec:quantum_algorithm_expectiles} and~\ref{sec:quantum_algorithm_rvar}. But before, we would like to mention that the quantum Fourier transform (QFT) and the controlled $Q^j$-gate operations (Grover operations) in~\prettyref{fig:amplitude_estimation_circuit} lead to deep circuits with a high number of CNOT-gates. This circumstance is a challenging task for noisy intermediate-scale quantum devices. As a result, the design of QAE variants that achieve Grover-type speed up without the use of QFT and the series of controlled $Q^j$-gates has become an active area of research~\parencite{aaronson_quantum_2019,grinko_iterative_2021,suzuki_amplitude_2020,nakaji_faster_2020}.
	In addition to the canonical QAE, we also include Maximum Likelihood QAE (MLQAE) as in~\textcite{suzuki_amplitude_2020} and Iterative QAE (IQAE) as in~\textcite{grinko_iterative_2021} in our analysis. 
	
	MLQAE does not use QFT and the controlled $Q^j$-gates. Instead, it uses different (non-controlled) powers of $Q$ to construct a sufficiently large sample statistic from which the desired amplitude is estimated.
	The repeated measurement of the state $Q^{k} \mathcal{A} \ket{0}$ can be interpreted as a Bernoulli process for which the unknown hitting probability $p_{k}$ can be converted into the searched amplitude. The parameters $p_{k}$ are estimated for exponentially increasing powers $k \in \{0,2,2^2,\ldots,2^{m-1}\}$ by the maximum likelihood method and combined to a final result.
	
	IQAE approximates the amplitude by constructing confidence intervals that contain the target parameter with a given probability.
	This is done in an iterative process which reduces the length of the  interval in each step until the desired approximation accuracy is achieved. As for the MLQAE, the QFT and the controlled $Q^j$-gates are not required. Instead, carefully chosen (non-controlled) $Q^k$-gates are used to construct tighter interval bounds.
	
	We compare these variants in the case study in~\prettyref{sec:numerical_study}. Further, all variants use the same operator $\mathcal{A}$ as input. Hence, we are allowed to restrict our attention until the end of this section to the construction of $\mathcal{A}$.\newline
	
	For the implementation on a quantum computer we assume that the support of a given random variable $X$ is the set $\{a,a+b,\dots,a+(2^{n}-1)b\}$ for some constants $a\in\realLine$, $b\in(0,\infty)$ and $n\in\naturalNumbers$. We describe this support by a quantum register $\ket{i}_n$ based on $n$ qubits. For every $i\in\{0,\dots,2^{n}-1\}$ we denote the probability that $X$ is equal to $a+bi$ by $p_i$. We define $S\defgl\{0,\dots,2^{n}-1\}$ and fix an arbitrary point $i^{*}\in S$. Further, define the map
	\begin{align*}
		g_X:S\rightarrow\left\{a,a+b,\dots,a+\left(2^{n}-1\right)b\right\},i\mapsto a+bi
	\end{align*}
	and set $x^{*}\defgl g_X(i^{*})$.
	
	\subsection{Operator $\mathcal{A}$ for expectiles}\label{sec:quantum_algorithm_expectiles}
	
	In this section, we fix a level $ \frac{1}{2}\leq\alpha<1$ and set $\beta=\frac{2\level-1}{1-\level}$.  
	
	To obtain the operator $\mathcal{A}$ in~\prettyref{fig:amplitude_estimation_circuit} we rewrite the integrand in~\prettyref{eq:auxiliary_function_expected_payoff} by defining the following function in dependence of $x$, representing a realization of $X$:
	\begin{align*}
		f_{x^{*}}(x)\defgl
		\begin{cases}
			x &,x<x^{*},\\
			x+\beta x-\beta x^{*} &,x\geq x^{*}.
		\end{cases}
	\end{align*}
	
	To obtain a map with domain $S$ we define for each $i\in S$:
	\begin{align*}
		f_{i^{*}}(i)\defgl (f_{x^{*}}\circ g_X)(i)=
		\begin{cases}
			bi + a &,i<i^{*},\\
			bi + a + \beta b i - \beta b i^{*} &,i\geq i^{*}.
		\end{cases}
	\end{align*}
	
	Following the notations in~\textcite{stamatopoulos_option_2020}, we set $f_{i^{*},\min}\defgl\min\limits_{i\in S}f_{i^{*}}(i)$ and $f_{i^{*},\max}\defgl\max\limits_{i\in S}f_{i^{*}}(i)$, as well as
	\begin{align*}
		\tilde{f}_{i^{*}}(i)\defgl 2\frac{f_{i^{*}}(i)-f_{i^{*},\min}}{f_{i^{*},\max}-f_{i^{*},\min}}-1.
	\end{align*}
	
	Then, for a scaling parameter $\gamma\in\left[0,1\right]$ we have 
	\begin{align*}
		\gamma\tilde{f}_{i^{*}}(i)+\frac{\pi}{4}=
		\begin{cases}
			g_{i^{*},0}(i) &,i<i^{*},\\
			g_{i^{*},0}(i)+g_{i^{*},1}(i)&,i\geq i^{*},
		\end{cases}
	\end{align*}
	with
	\begin{align*}
		&g_{i^{*},0}(i)\defgl  2\gamma\left(\frac{1}{(1+\beta)(2^{n}-1)-\beta  i^{*}}\right)i-\gamma+\frac{\pi}{4},\\
		&g_{i^{*},1}(i)\defgl  2\gamma\left(\frac{\beta }{(1+\beta)(2^{n}-1)-\beta  i^{*}}\right)i-2\gamma\left(\frac{\beta }{(1+\beta)(2^{n}-1)-\beta  i^{*}}\right)i^{*}.
	\end{align*}
	
	The operator $\mathcal{A}$ is constructed in the following manner: First, load the distribution of $X$ into the $\ket{i}_{n}$ register. Secondly, perform a comparator such that the qubit $\ket{c}$ is in state $\ket{1}$ if $i\geq i^{*}$ or $\ket{0}$ if $i<i^{*}$. Finally, perform the multi-controlled y-rotations illustrated in~\prettyref{fig:multi_controlled_y_rotation}, in which a single-qubit y-rotation with respect to an angle $\theta$ is given by the following unitary matrix:
	\begin{align*}
		R_{y}(\theta)\defgl
		\begin{pmatrix}
			\cos(\theta/ 2) & -\sin(\theta/ 2)\\
			\sin(\theta/ 2) & \cos(\theta/ 2)
		\end{pmatrix}.
	\end{align*}
	
	\begin{figure}[!ht]
		\[
		\begin{array}{c}
			\Qcircuit @C=1.5em @R=1.2em {
				&\lstick{\ket{i}_{n}}& \ctrl{2} & \ctrl{1} & \qw \\
				&\lstick{\ket{c}_{\,\,\,}}& \qw & \ctrl{1} & \qw \\
				&\lstick{\ket{0}_{\ }}& \gate{R_y(2g_{i^{*},0}(i))} & \gate{R_y(2g_{i^{*},1}(i))} & \qw 
			}
		\end{array}
		\]
		\caption{Circuit of multi-controlled y-rotations to describe the payoff function.}
		\label{fig:multi_controlled_y_rotation}
	\end{figure}
	
	The operator $\mathcal{A}$ maps the initial state $\ket{0}_{n}\ket{0}\ket{0}$ of the $n+2$ qubits to the following state:
	\begin{align*}
		&\sum\limits_{i<i^{*}}\sqrt{p_i}\ket{i}_{n}\ket{0}\Big(\cos\big(g_{i^{*},0}(i)\big)\ket{0}+\sin\big(g_{i^{*},0}(i)\big)\ket{1}\Big)\\
		&\quad +\sum\limits_{i\geq i^{*}}\sqrt{p_i}\ket{i}_{n}\ket{1}\Big(\cos\big(g_{i^{*},0}(i)+g_{i^{*},1}(i)\big)\ket{0}+\sin\big(g_{i^{*},0}(i)+g_{i^{*},1}(i)\big)\ket{1}\Big).
	\end{align*}
	
	After applying the amplitude estimation to this operator we obtain an estimation of the probability that the last qubit is in state $\ket{1}$. This probability is given by
	\begin{align*}
		\sum\limits_{i<i^{*}}p_i\Big(\sin\big(g_{i^{*},0}(i)\big)\Big)^{2} + \sum\limits_{i\geq i^{*}}p_i\Big(\sin\big(g_{i^{*},0}(i)+g_{i^{*},1}(i)\big)\Big)^{2}.
	\end{align*}
	
	Applying the approximation
	\begin{align*}
		\Bigg(\sin\left(\gamma\tilde{f}_{i^{*}}(i)+\frac{\pi}{4}\right)\Bigg)^{2} \approx \gamma\tilde{f}_{i^{*}}(i)+\frac{1}{2},
	\end{align*}
	to this probability, we obtain an estimator for an affine transformation of $h_{X,\level}\left(x\right)$.
	
	\begin{exam}[Normal distribution]\label{exam:normal}
		We assume a normal distributed random variable $Y$ with mean $\mu\in\realLine$ and standard deviation $\sigma\in(0,\infty)$. We denote by $\varphi$, respectively $\Phi$, the probability density function, respectively the cumulative distribution function, of a standard normal distributed random variable. With this example, we would like to demonstrate the concrete form of the map in~\prettyref{eq:auxiliary_function_expected_payoff}, which is given for each $y\in\realLine$ by
		\begin{align*}
			h_{Y,\level}(y)=\mu + \beta\left(1-\Phi\left(\frac{y-\mu}{\sigma}\right)\right)(\mu-y) + \beta\sigma\varphi\left(\frac{y-\mu}{\sigma}\right).
		\end{align*}
		To obtain an estimator for this expression we use a random variable $X$ with support given by $\left\{\mu-3\sigma,\mu-3\sigma +\frac{6\sigma}{2^{n}-1},\dots,\mu+3\sigma\right\}$ and probabilities obtained from the distribution of $Y$. The function $f_{i^{*}}$ is defined by the parameters $a=\mu-3\sigma$ and $b=\frac{6\sigma}{2^{n}-1}$. Hence, it holds that
		\begin{align*}
			f_{i^{*},\min}=\mu-3\sigma \quad \text{and}\quad f_{i^{*},\max}=\mu+3\sigma+6\beta\sigma\left(1-\frac{i^{*}}{2^{n}-1}\right).
		\end{align*}
	\end{exam}
	
	\subsection{Operator $\mathcal{A}$ for Range Value-at-Risk}\label{sec:quantum_algorithm_rvar}
	
	The basic idea for calculating the RVaR is to combine two comparator circuits and use them to control the application of an appropriate $y$-rotation similiar to the one in ~\prettyref{fig:multi_controlled_y_rotation}. This is then combined with the calculation of VaR-values as shown in Algorithm 2. 
	
	To define the y-rotation we proceed analogously to~\prettyref{sec:quantum_algorithm_expectiles}.
	Given the affine mapping $g:=g_X$ with $g_{\min} := \min\limits_{i \in S} g(i)$ and $g_{\max} := \max\limits_{i \in S} g(i)$, we define 
	\begin{align*}
		\tilde{g}(i) := 2 \frac{g(i)-g_{\min}}{g_{\max}-g_{\min}} - 1.
	\end{align*}
	
	For the scaling parameter $\gamma \in [0,1]$ we set
	\begin{align*}
		\hat{g}(i) := \gamma \tilde{g}(i) + \frac{\pi}{4}.
	\end{align*}
	The $y$-rotations are then given by $R_y(2\hat{g}(i))$. 
	
	To explain the simultaneous application of two different comparators in more detail, we extend the notations from~\prettyref{sec:quantum_algorithm_expectiles}. For $k \in \mathbb{Z}$, let $\mathrm{cmp}_{1}(k)$ and $\mathrm{cmp}_{2}(k)$ be the unitary operators defined by the following rule: 
	\begin{align*}
		&\mathrm{cmp}_{1}(k) : \ket{i}_{n}\ket{0}_{\ } \longmapsto 
		\begin{cases} 
			\ket{i}_{n}\ket{1} &,i \geq k, \\ 
			\ket{i}_{n}\ket{0} &,i < k,
		\end{cases}\\
		&\mathrm{cmp}_{2}(k) : \ket{i}_{n}\ket{0}_{\ } \longmapsto 
		\begin{cases} 
			\ket{i}_{n}\ket{1} &,i < k, \\ 
			\ket{i}_{n}\ket{0} &,i \geq k. 
		\end{cases}	
	\end{align*}
	
	These comparator circuits use the same $n$-qubit register $\ket{i}_{n}$ to store the binary representation of $i$. Furthermore, comparator circuit $j$ uses an  $(n-1)$-qubit register $\ket{a_{j}}$ for ancillas and one qubit $\ket{t_j}$ to save the result of the comparison, see~\prettyref{fig:cmp}.
	\begin{figure}[!ht]
		\[
		\begin{array}{c}
			\Qcircuit @C=1.5em @R=0.75em {
				&\lstick{\ket{a_1}} & \multigate{2}{\mathrm{cmp_1}(k_1)} & \qw  & \qw & \qw\\
				&\lstick{\ket{t_1}_{\,}} &\ghost{\mathrm{cmp_1}(k_1)}  & \qw & \ctrl{4} & \qw\\
				&\lstick{\ket{i}_{n\,}}  &\ghost{\mathrm{cmp_1}(k_1)} & \multigate{2}{\mathrm{cmp_2}(k_2)}& \qw& \qw\\
				&\lstick{\ket{t_2}_{\,}} & \qw & \ghost{\mathrm{cmp_2}(k_2)} & \ctrl{2} & \qw\\
				&\lstick{\ket{a_2}} & \qw & \ghost{\mathrm{cmp_2}(k_2)} & \qw& \qw \\
				&\lstick{\ket{0}_{\ }} & \qw & \qw & \gate{R_y(2\hat{g}(i))} & \qw\\
			}
		\end{array}
		\]
		\caption{Quantum circuit for calculating $\rangeValueAtRisk{\alpha,\beta}{X}$. The distribution of $X$ is loaded into $\ket{i}_{n}$. The ancillas for the comparator circuits are stored in $\ket{a_1}$ and $\ket{a_2}$. Comparator results are placed in $\ket{t_1}$ and $\ket{t_2}$, which are used to control the $y$-rotation $R_y(2\hat{g}(i))$.}
		\label{fig:cmp}
	\end{figure}
	
	We combine the two comparator circuits in the sense that they constrain $i$ from both sides and the two result qubits control the $y$-rotation. Finally, the operator $\mathcal{A}$ is created by loading the distribution into the qubit register $\ket{i}_{n}$. Note that at this point $\mathcal{A}$ still depends on the choice of $k_1$ and $k_2$.
	
	Suppose $k_2 < k_1$. If we neglect the ancillas and apply $\mathcal{A}$ to the initial state $\ket{0}_{n}\ket{0}\ket{0}\ket{0}$, where the second and third qubit represent $t_1$ and $t_2$, we obtain
	\begin{align*}
		\mathcal{A}\ket{0}_{n}\ket{0}\ket{0}\ket{0} &= \sum_{i=0}^{k_2-1} \sqrt{p_i}\ket{i}_{n}\ket{1}\ket{0}\ket{0}\\
		&\quad+ \sum_{i=k_2}^{k_1} \sqrt{p_i}\ket{i}_{n}\ket{1}\ket{1} \Big( \cos(\hat{g}(i)) \ket{0} + \sin(\hat{g}(i))\ket{1} \Big)\\
		&\quad+ \sum_{i=k_1+1}^{2^n-1} \sqrt{p_i}\ket{i}_{n}\ket{0}\ket{1}\ket{0}.
	\end{align*}
	
	The probability that this state is $\ket{1}$ at the position of the last qubit is $\sum_{i=k_2}^{k_1}p_i \sin^2(\hat{g}(i))$.
	
	If we apply the approximation $\sin^2(\hat{g}(i)) \approx \gamma \tilde{g}(i) + \frac{\pi}{4}$ as in~\prettyref{sec:quantum_algorithm_expectiles} and choose for $k_1$ and $k_2$ integer approximations for $\valueAtRisk{\alpha}{g^{-1} \circ X}$ and $ \valueAtRisk{\beta}{g^{-1} \circ X}$, then we obtain
	\begin{align*}
		\sum_{i=k_2}^{k_1}p_i \sin^2(\hat{g}(i)) \approx  \frac{-2\gamma P(k_2 \leq X \leq k_1)\rangeValueAtRisk{\alpha,\beta}{g^{-1}\circ X} - g_{\min}}{g_{\max}-g_{\min}} - \gamma + \frac{1}{2}.
	\end{align*}
	
	From this expression we can reconstruct $\rangeValueAtRisk{\alpha,\beta}{X}$. Note that integer approximations for 
	$\valueAtRisk{\alpha}{g^{-1} \circ X}$ and $ \valueAtRisk{\beta}{g^{-1} \circ X}$ and an estimate for $P(k_2 \leq X \leq k_1)$
	can be obtained with the quantum methods in~\textcite{woerner_quantum_2019}.

	\section{Two case studies}\label{sec:numerical_study}
	
	In this section, we analyze the computation of VaR, CVaR, expectiles and RVaR on noiseless and noisy simulators.
	The noise model for the noisy simulator mimics the behavior of the quantum computer IBM Brisbane.
	All computations are performed with the help of the qiskit framework, see~\parencite{treinish_qiskitqiskit_2022}. 
	
	\subsection{Case study 1: Credit risk}
	
	In this section, we demonstrate the evaluation of risk measures using an example 
	from the financial industry and compare the performance between the algorithms for RVAR, expectiles and VaR. We consider the calculation of economic capital requirements as carried out by banks 
	in the context of credit risk assessment, cf.~\textcite{rutkowski}. We also discuss the effects of scaling the number of qubits to a size that is realistic in practice, see~\prettyref{rem:scalingRealWorldProblem} below.

	By legal requirements, banks have to assess the credit default risk of their 
	portfolios and have to provide sufficient capital reserves in order to cover potential losses, cf.~\textcite{basel1,basel2,basel3}.
	In the following, we briefly present a widely used model for assessing credit risk, cf.~\textcite{rutkowski} and ~\textcite{egger_credit_2021}:

	Assume that the bank's portfolio consists of $K$ assets and let $k \in \{1,...,K\}$. The loss of the $k$-th asset is modeled by the random variable $L_k = \lambda_kX_k$,
	where $X_k$ is a Bernoulli random variable and $\lambda_k$ is a positive real number 
	that quantifies the amount of loss in the event of default. 
	The total loss of the bank is defined as the random variable $L=L_1+...+L_K$.

	In practice, credit defaults on different assets are correlated, for example because they are affected by the same economic event.
	The correlation is modeled by making the default probabilities 
	of the $X_k$ dependent on the realization of a common, latent random variable $Z$.
	More precisely, the default probabilities of the Bernoulli random variables $X_k$ are given by
	\begin{align*}
		p_k(z) = F\left(\frac{F^{-1}(p_k^0)-\sqrt{\rho_k}z}{\sqrt{1-\rho_k}}\right),
	\end{align*}
	where $z \in \mathbb{R}$ denotes a realization of a standard normal distributed random variable $Z$, $p_k^0 \in [0,1]$,
	$F$ being the cumulative distribution function of $Z$ and $\rho_k\in[0,1)$ is a sensitivity parameter. 
	Now, the capital reserve is calculated by evaluating risk measures of the total loss $L$.

	A quantum circuit for loading the distribution of $L$ is presented in \textcite{egger_credit_2021}.
	In the upcoming example of a credit portfolio, the calculation of the capital 
	reserve using VaR and expectile gives us an almost identical reserve, while the result using RVaR is significantly different. By varying the parameter $m$ in the QAE (recall Section~\ref{sec:qae}), we compare the accuracy of the algorithms in~\prettyref{fig:credit_risk_example}. 

	We assume that the credit portfolio consists of two assets with model parameters given in~\prettyref{tab:parameters_example}.

	\begin{table}[!ht]
		\centering
		\begin{tabular}{l|c|c|c|c}
			$k$ & $p_k^0$ & $\rho_k$ &  $\lambda_k$\\
			\hline
			1 & 0.12 & 0.1 & 1  \\
			2 & 0.35 & 0.05 & 2 
		\end{tabular}
		\caption{Parameters of the credit risk portfolio.}
		\label{tab:parameters_example}
	\end{table}

	The standard normal distribution is stored in a register with $2$ qubits. 
	The density of the normal distribution is restricted to the interval $[-2,2]$.
	The loading of the distribution of $L$ can be divided into two steps, see \prettyref{fig:circuit_loss}. 
	The first step consists of constructing the operator $U$, the second of constructing the operator $S$.

	\begin{figure}[!ht]
		\[
		\begin{array}{c}
			\Qcircuit @C=1.5em @R=0.75em {
				\lstick{\ket{Z_1}} & \multigate{3}{U} & \qw & \meter \\
				\lstick{\ket{Z_2}} & \ghost{U} & \qw & \meter\\
				\lstick{\ket{X_1}} & \ghost{U} & \multigate{3}{S} & \meter \\
				\lstick{\ket{X_2}} & \ghost{U} & \ghost{S}& \meter \\
				\lstick{\ket{L_1}} & \qw & \ghost{S}& \qw \\
				\lstick{\ket{L_2}} & \qw & \ghost{S}& \qw 
			}
		\end{array}
		\]
		\caption{Quantum circuit for loading the distribution of $L$. Qubits $Z_1$ and $Z_2$ represent the discretized normal distributed random variable $Z$, qubits $X_1$ and $X_2$
		the Bernoulli random variables of each asset and qubits $L_1$ and $L_2$
		the losses of each asset.}
		\label{fig:circuit_loss}
	\end{figure}

	In the first step, the distribution of the multivariate random variable 
	$(X_1,X_2)$ is loaded into a quantum register. This is done with the circuit in \prettyref{fig:circuit_U}.

	\begin{figure}[!ht]
		\[
		\begin{array}{c}
			\Qcircuit @C=1.5em @R=0.75em {
				\lstick{\ket{Z_1}} & \multigate{1}{U_Z} & \qw & \ctrl{2} & \qw & \ctrl{3} & \qw \\
				\lstick{\ket{Z_2}} & \ghost{U_Z} & \ctrl{1} & \qw & \ctrl{2} & \qw & \qw \\
				\lstick{\ket{X_1}} & \gate{R_y(b_1)} & \gate{R_y(a_1)} & \gate{R_y(2a_1)} & \qw & \qw & \qw \\
				\lstick{\ket{X_2}} & \gate{R_y(b_2)} & \qw & \qw & \gate{R_y(a_2)} & \gate{R_y(2a_2)} & \qw 
			}
		\end{array}
		\]
		\caption{Quantum circuit for the operator $U$. Here, $U_Z$ is the
		circuit loading the normal distribution on $2$ qubits. The angles $a_k, b_k$ arise from 
		a first-order-approximation of $2 \arcsin(\sqrt{p_k(z)})$.}
		\label{fig:circuit_U}
	\end{figure}

	The values $a_k, b_k$ are obtained by a first-order-approximation of the rotation angles 
	\begin{align*}
		\Theta_p^k(z) = 2 \arcsin(\sqrt{p_k(z)}).
	\end{align*}
	The values of $z$ are given as $z_i = a_z i + b_z$, where $i \in \{0,1,2,3\}$ and $a_z = 4/3,\ b_z = -2$. Here, $i \mapsto \tfrac{4}{3}i-2$ is the affine 
	transformation mapping $[0,3]$ onto $[-2,2]$.
	By the formula for $p_k(z)$, we have for the first derivative
	\begin{align*}
		p^{\prime}_k(z) = - \sqrt{\frac{\rho_k}{1-\rho_k}} F^{\prime}\left( \frac{F^{-1}(p_k^0)-\sqrt{\rho_k}z}{\sqrt{1-\rho_k}}\right).
	\end{align*}
	Since $F^{\prime}$ is the density of the normal distribution, we obtain for the derivative of $\Theta_p^k$ that 
	\begin{align*}
		(\Theta_p^k)^{\prime}(z) = \frac{p^{\prime}_k(z)}{\sqrt{p_k(z)}} \frac{1}{\sqrt{1-p_k(z)}}.
	\end{align*}
	The first-order-approximation of $\Theta_p^k$ around $0$ is $\Theta_p^k(z) \approx \Theta_p^k(0) + (\Theta_p^k)^{\prime}(0) z$. For $z_i = a_z i + b_z$ we get $\Theta_p^k(z_i) \approx b_k + a_k i$ with $b_k:=\Theta_p^k(0)+b_z$ and $a_k = (\Theta_p^k)^{\prime}(0)a_z$.
	These quantities are summarized in \prettyref{tab:parameters_creditRisk}.
	\begin{table}[!ht]
		\centering 
		\begin{tabular}{l|c|c|c|c}
			$k$ & $\Theta_p^k(0)$ & $(\Theta_p^k)^{\prime}(0)$  & $a_k$ & $b_k$\\
			\hline
			1 & 0.7592 & -0.2127 & -0.2836 & -1.2408 \\
			2 & 1.0344 & -0.1676 & -0.2234 & -0.9656
		\end{tabular}
		\caption{Parameters for loading the credit risk portfolio onto a quantum computer.}
		\label{tab:parameters_creditRisk}
	\end{table}
	
	The operator $S$ calculates the sum of the random variables $L_k$. It operates on the $X$- and $L$-register 
	as follows:
	\begin{align*}
		S: \ket{X_1}\ket{X_2}\ket{0}\ket{0} \longmapsto \ket{X_1}\ket{X_2}\ket{\lambda_1X_1+\lambda_2X_2}.
	\end{align*}
	Its exact implementation can be found in detail in~\textcite{stamatopoulos_option_2020}.
	\prettyref{fig:dist_loss} shows the distribution of $L$ that we obtain in our example.

	\begin{figure}[!ht]
		\centering
		\includegraphics[scale=0.55]{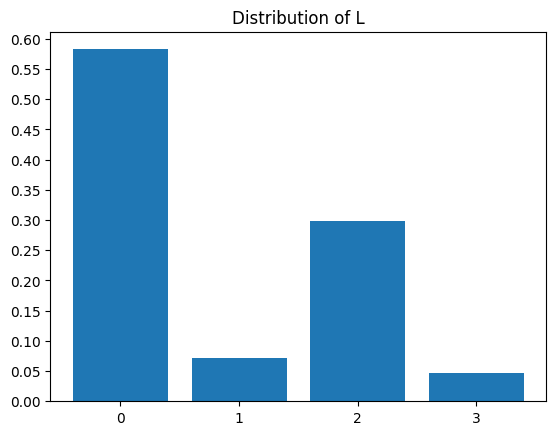}
		\caption{Distribution of the random variable $L$ given the parameters in 
		\prettyref{tab:parameters_example}.} 
		\label{fig:dist_loss}
	\end{figure}

	\prettyref{fig:credit_risk_example} shows the average result obtained 
	by calculating risk measures of $L$ when QAE is executed multiple times.
	It shows that 
	the average result converges to the exact result if the number $m$ of ancilla qubits increases. 
	In addition, the figure shows that the accuracy is similar in the case of expectile and VaR, while the accuracy in the case of RVaR is less accurate.

	\begin{figure}[!ht]
		\centering
		\includegraphics[scale=0.45]{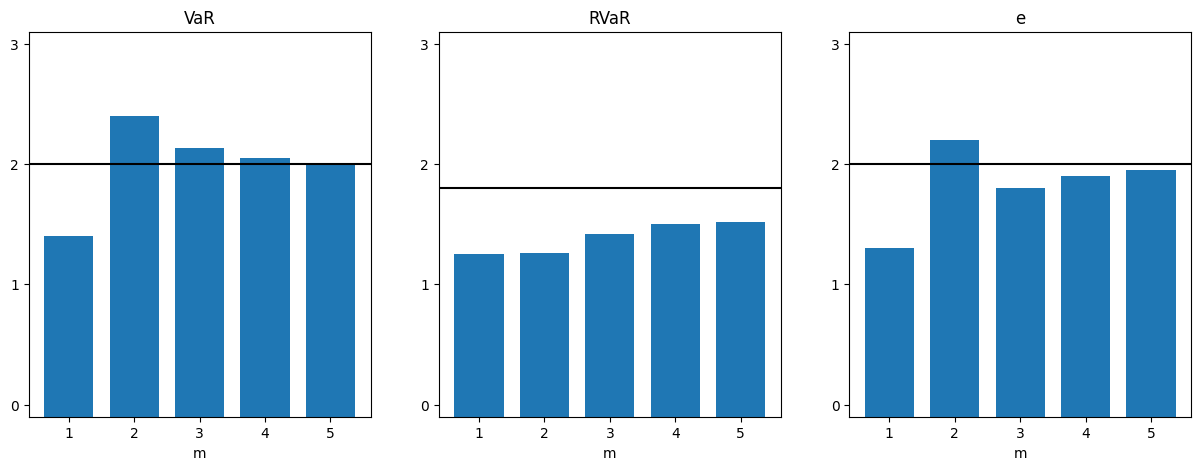}
		\caption{The average result of calculating the risk measure of $L$ when QAE is executed $1000$ times. 
		The parameter $m$ indicates the number of ancilla qubits used in the QAE. 
		The black horizontal line is the exact value of the corresponding risk measure.}
		\label{fig:credit_risk_example}
	\end{figure}

	\begin{rem}\label{rem:scalingRealWorldProblem}
		In case of the VaR, \textcite[Section V]{egger_credit_2021} state an argumentation how the aforementioned problem scales to a size that is relevant for the financial industry. They analyzed
		the circuit depth (number of Toffoli-gates) in dependence of the number of assets $K$ (or more precisly the number of qubits $n_S$ used to represent the total loss). The number of Toffoli-gates as a measure for the complexity is, e.g.,~also used in~\textcite{chakrabarti_threshold_2021} . \textcite[Table II]{egger_credit_2021} states the number of Toffoli-gates for the operators $U$ and $S$, which do not depend on the VaR-algorithm itself. Hence, we obtain the same numbers in case of expectiles or RVaR. The choice of a risk measure is part of operator $\mathcal{C}$ in~\textcite{egger_credit_2021}, for which~\textcite[Table II]{egger_credit_2021} gives us the value $2\lfloor\log_2(n_S-1)\rfloor+9$. 
		This value stems from~\textcite{draper2006} and it is the number of Tofolli-gates for a fixed value comparator. Since, the operator in case of expectiles is also based on a fixed value comparator,	this number can be adopted to the expectile case. Also the calculation for RVaR is based on two fixed comparator circuits, compare with~\prettyref{fig:cmp}. Hence, this value can be scaled linearly and therefore the value $2\lfloor\log_2(n_S-1)\rfloor+9$ acts also as an orientation in the RVaR-case.
	\end{rem}

	\subsection{Case study 2: Insurance claim size data}
	
	In the second case study, we calculate the risk measures based on real world data fitted to a lognormal and a gamma distribution, because of their popularity in different disciplines. For instance, the lognormal distribution is applied in the description of growth processes~\parencite{sutton_gibrats_1997, huxley_problems_1993} as well as for modeling stock prices in the Black-Scholes model~\parencite{black_pricing_1973}. The gamma distribution is widely used in actuarial science to model the claim size distribution for non-life insurance contracts~\parencite{boland_statistical_2007, ohlsson_non-life_2010, laudage_severity_2019}. It is also applied for failure-time analysis~\parencite{scheiner_design_2001}. To do so, we denote by $LN(\mu,\sigma)$ a lognormal distribution with expectation $\exp(\mu+\frac{\sigma^2}{2})$ and by $\Gamma(p,q)$ a gamma distribution with mean $p/q$.
	
	To obtain realistic parameter values, 
	we calibrate the parameters towards a real world data set, namely the Norwegian automobile data set from the 
	\texttt{R}-package \texttt{CASdatasets}. This data set is, e.g.,~analyzed in~\textcite{fung_soft_2024}. 
	As done in there, we delete the values 1, 99 and 16999 from the data set. Furthermore, 
	insurance data often admits heavy tails. 
	Since this is not the main purpose in our case, we only focus on values below 100000. 
	This gives us 7704 data points in total. 
	We estimate the parameters of the lognormal and gamma distributions via the method of moments. 
	For the formulas in case of the lognormal distribution, see e.g.,~\textcite{ginos_parameter_2009}. 
	Additionally, we plot in Figure~\ref{fig:hist_norauto} the density functions of the calibrated lognormal and gamma distribution, 
	as well as, the histogram of the data.
		
	\begin{table}[!ht]
		\centering
		\begin{tabular}{l|c|c}
			Distribution of $-Y$ & Parameters & Interval \\
			\hline
			Lognormal $LN(\mu,\sigma^2)$ & $\mu = 9.6754$, $\sigma =  0.7416$ & $(0,100000]$ \\
			Gamma $\Gamma(p,q)$& $p = 1.3635$, $\frac{1}{q} = 15373$ & $(0,100000]$ \\
		\end{tabular}
		\caption{Parameters of the applied distributions and the intervals for their discretization.}
		\label{tab:parameters_norauto}
	\end{table}

	\begin{figure}[!ht]
		\centering
		\includegraphics[scale=0.45]{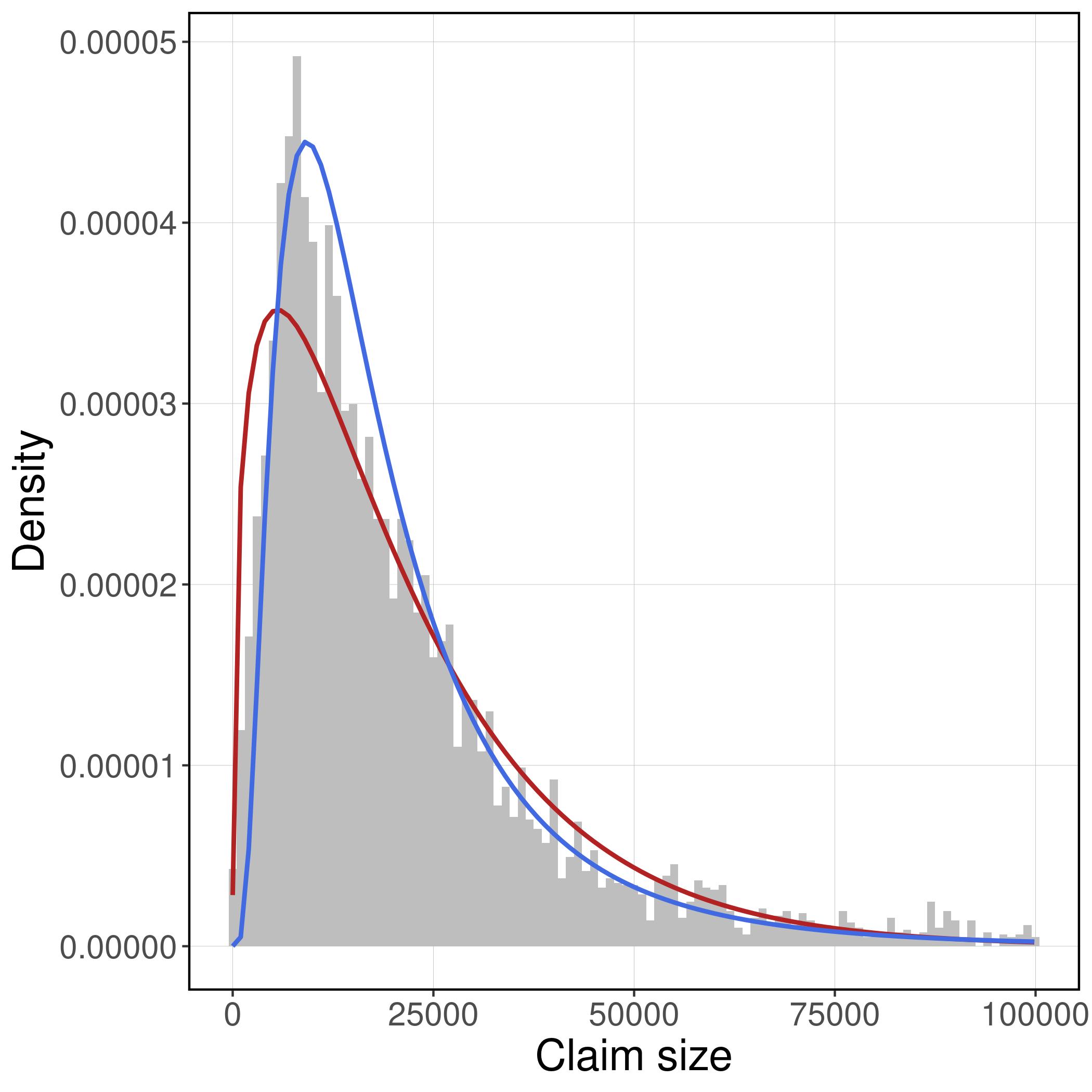}
		\caption{Histogram of the Norwegian automobile data set and the calibrated lognormal (blue curve) and gamma (red curve) distribution.}
		\label{fig:hist_norauto}
	\end{figure}

	Our analysis follows the idea behind~\prettyref{exam:normal}. As in there, let $Y$ denote a continuously distributed random variable with density function $f$.
	As a first step, we would like to express the values of $f$ as norm-squared amplitudes of a suitable quantum state.
	Given that we have $n$ qubits available, we need to restrict our attention to a bounded interval included in the domain of $f$ and discretize this interval by $2^n$ points. 
	Formally, we move from the continuous random variable $Y$ to a discretized version of it, denoted by $X$. 
	That means also that we replace $f$ with the density function of $X$. 
	The quantum state will then capture the probabilities of $g_X^{-1} \circ X$, 
	where $g_X$ is the affine mapping that transforms $\{0,1,\ldots,N-1\}$ to the support of $X$. 
	\prettyref{tab:parameters_norauto} summarizes the parameters for the  distributions of $-Y$ (lognormal, gamma) 
	and the interval to which the associated density function is restricted. 
	
	There are two effects that have a major impact on the quality of the calculation of risk measures on a quantum computer. One is the coherence time and gate fidelity, i.e.,~the computational accuracy of the quantum computer. The other is the accuracy of the approximation of the continuous density function $f$ by the probability mass function of~$g_X^{-1} \circ X$. 
	
	The latter depends on the number of qubits for loading the distribution, because it determines the granularity of the discretization of the domain of $f$. \prettyref{fig:dist_gamma} shows the influence of hardware noise and qubit count on the distribution.
	\begin{figure}[!ht]
		\centering
		\includegraphics[scale=0.210]{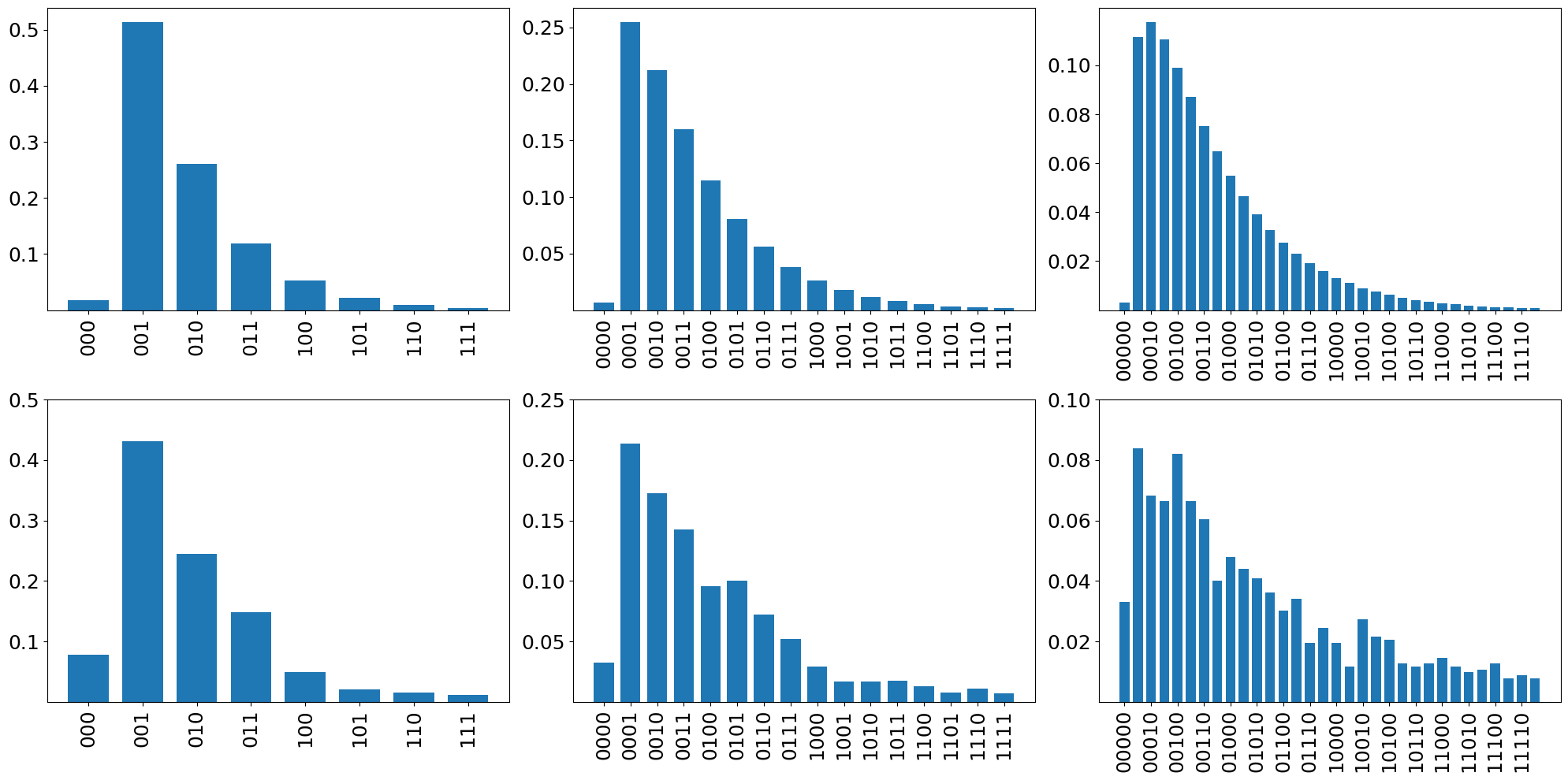}
		\caption{Distribution loading for the gamma distribution 
		$\Gamma(1.3635,15373^{-1})$ with respect to different numbers of qubits. 
		The blue bars describe the quantum state generated by the distribution circuit, 
		or equivalently, the probability mass function of $g_X^{-1} \circ X$. 
		The first row shows the results of a noiseless, the second row the results of a noisy simulator. 
		From the left to the right column, the number of qubits is 3, 4 and 5.}
		\label{fig:dist_gamma}
	\end{figure}

	We see that hardware noise leads to inaccuracies in the description of the density function. This is typical for noisy intermediate-scale quantum devices.
	
	The granularity of the discretization is of particular interest for the calculation of risk measures, because it can lead to noticeable changes in the capital requirements. To illustrate this statement, we compare the results between the true value of a risk measure and its value calculated on a quantum simulator in~\prettyref{fig:sim_result}.
	
	Here, we use for VaR and CVaR a level of $\lambda = 0.05$. 
	Further, for the expectile we use the parameter value $\alpha = 0.05$  and for the RVaR the parameters 
	$\alpha = 0.005$ and $\beta = 0.05$. For the definitions of these risk measures and the corresponding meaning behind $\alpha$, $\beta$ and $\lambda$, the reader should recall~\prettyref{sec:definition_risk_measures}.
	\prettyref{fig:sim_result} states that good results are produced 
	if the number of qubits 
	is sufficiently large.
	For all risk measures, the approximation error (relatively to the length of the interval to discretize the domain of the density 
	function) becomes sufficiently small, i.e.,~smaller than $0.025$ on the y-axis in~\prettyref{fig:sim_result}, 
	if the number of qubits for loading the distribution increases.
	
	It is worth noting that the parameter $\gamma$, which appears in the construction of the operator $\mathcal{A}$, has a significant impact on the accuracy of the results. In general, a suitable choice for $\gamma$ depends on the distribution and the risk measure under consideration. As a rule of thumb for our examples, we recommend a value of $\gamma \approx \pi/8$ for VaR and CVaR, and a value of $\gamma \approx \pi/4$ for RVaR and expectiles. These values arise from the heuristical experience of the authors during creating the presented examples.
	
	\begin{figure}[!ht]
		\centering
		\includegraphics[scale=0.4]{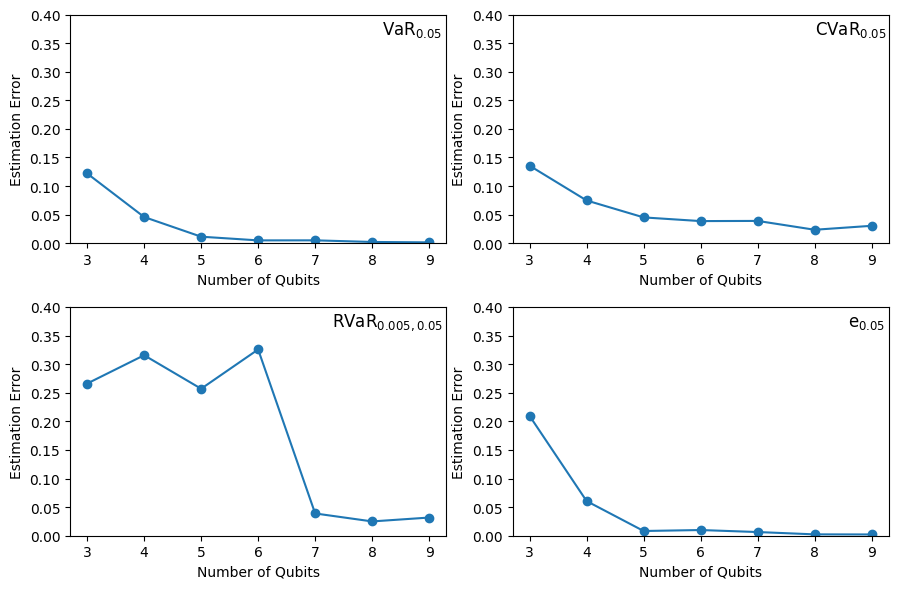}
		\caption{The estimation error on a simulator as a function of the number of 
		qubits used for loading the distribution. 
		The error is given relatively to the length of the domain on which $f$ is defined. 
		For the results of this plot we applied the IQAE with target precision $0.05$ and confidence level 
		$0.01$ to $\Gamma(1.3635,15373^{-1})$. We used $\gamma = \pi/8$ for VaR and CVaR and $\gamma = \pi/4$ for RVaR and expectiles. An analogous error behavior is given for the lognormal distribution, see Appendix~\ref{sec:figures}.}
		\label{fig:sim_result}
	\end{figure}
	
	Besides the distribution loading, the number of ancilla qubits is also responsible for the approximation accuracy in the canonical QAE, because it determines the amount of classical bits used to approximate the result. In the case of an expectile, \prettyref{fig:canonical_gamma} shows how an increase in the number of ancilla qubits (parameter $m$ in~\prettyref{fig:amplitude_estimation_circuit}) allows for an improvement in the accuracy of the results. This plot is motivated by Figure~4 in~\textcite{woerner_quantum_2019}.
	\begin{figure}[!ht]
		\centering
		\includegraphics[scale=0.45]{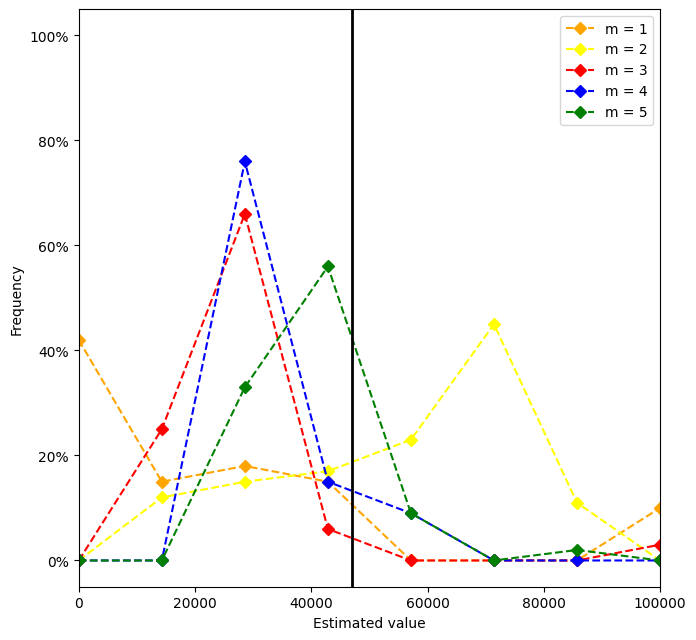}
		\caption{Estimation for the expectile calculated with the canonical QAE on a simulator. 
		We used $\gamma = \pi/4$. The underlying distribution is $\Gamma(1.3635,15373^{-1})$. 
		The exact value is indicated by the black line. With an increased number of ancilla qubits $m$, the most frequent result approaches the exact value. An analogous behavior is given for the lognormal distribution, see Appendix~\ref{sec:figures}.}
		\label{fig:canonical_gamma}
	\end{figure}

	\prettyref{fig:qpu_results} illustrates the performance of our algorithms as it would be expected on a real quantum device. 
	Since calculations with many qubits on real hardware suffer strongly from noise, we restrict ourselves 
	to three qubits when loading the distribution on a noisy simulator. The latter emulates a real device.
	Accordingly, we choose levels $\lambda = \alpha = 0.20$ for VaR, CVaR and EVaR and $\alpha = 0.05$, $\beta = 0.2$ 
	for RVaR.
	Also, we restrict ourselves to $m = 3$ qubits for the canonical QAE. 
	The algorithms are applied to different numbers of shots, which allows us to test potential improvements in accuracy for IQAE and MLQAE. 
	In general, the accuracy of the canonical QAE does not change with an increased number of shots, 
	because the canonical QAE leads to the most frequent result, compare with the third column in~\prettyref{fig:qpu_results}.
	\begin{figure}[!ht]
		\centering
		\includegraphics[scale=0.5]{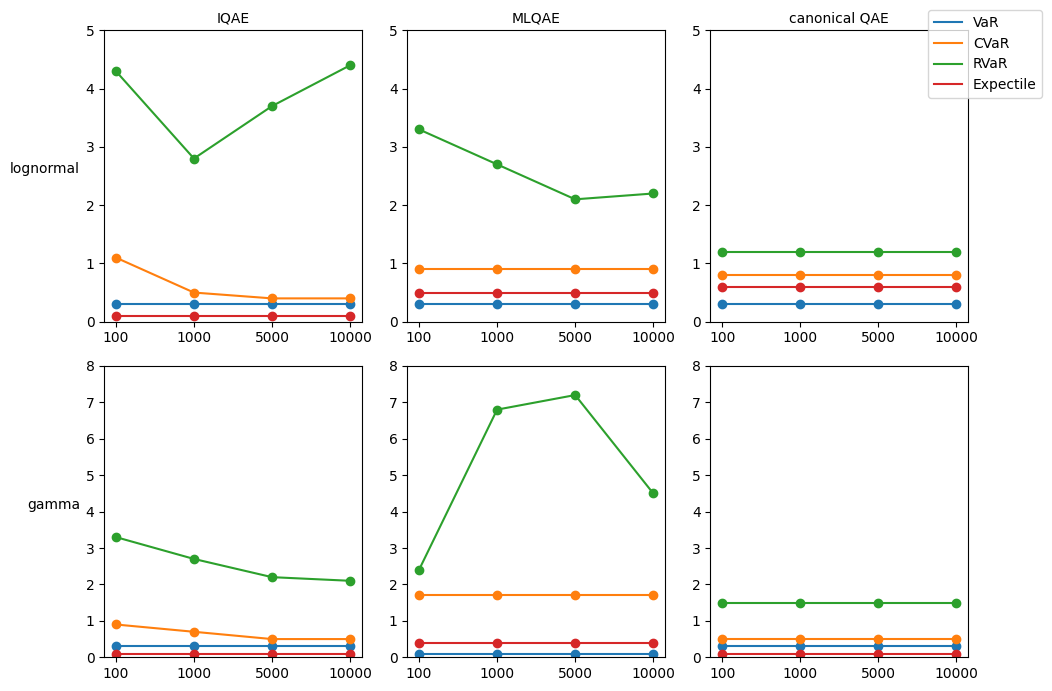}
		\caption{The estimation error on a noisy simulator as a function of the number of shots. 
		The error is given relatively to the length of the interval for discretizing the domain of $f$, 
		i.e.,~the estimation error is given as outcome of the amplitude estimation minus 
		the true value and this difference divided by the length of the mentioned interval.
		IQAE was applied with a confidence level of $0.05$, MLQAE with a evaluation schedule of $3$ 
		and canonical QAE with $m = 3$.
		Unlike RVaR, the expectile calculations, with the chosen algorithms and with $3$ qubits for the 
		distribution, are sufficiently reliable in the presence of noise.}
		\label{fig:qpu_results}
	\end{figure}
	
	The minimum estimation error in \prettyref{fig:qpu_results} in the case of the gamma distribution and  
	expectiles is $0.125$ for IQAE and canonical QAE. In contrast, 
	for MLQAE we obtain a significantly larger minimum estimation error of $0.375$. 
	In addition, note that the accuracy for VaR and expectiles is substantially better 
	than for CVaR and RVaR. There are two main reasons leading to these observations:
	
	First, VaR and expectiles are not direct results of the outcome of the amplitude estimation.
	The latter is only used to select the next subinterval in the bisection search algorithm. 
	Since the selection of the next subinterval is based on a comparison (with respect to the $\leq$-relation) 
	between the amplitude estimation result and a constant, 
	the correct subinterval is selected even if there are small errors in the amplitude estimation. The only case in which this fails is the one in which the QAE gives a wrong sign for the used objective function $h_{X,\alpha}$ in Algorithm 1. But, since it is more probable that this occurs for later iterations of the algorithm, in which the interval length is already small, the errors are also small. Hence, the main influence on the estimation error (if the tolerance values $\epsilon,\delta$ in Algorithm 1 are small enough) stems from the discretization of the domain of $f$, i.e.,~in the majority of the cases the estimated value is the closest possible to the true value regarding the chosen discretization. In this sense, the algorithms for VaR and expectiles are often close to the exact results, which means that the influence of the noise of the quantum hardware is reduced by this technique.
    
    Second, regarding the results of the CVaR and the RVaR, the QAE is not only used to compute the next subinterval, 
	but it is also applied in a second round to compute a conditional expectation. 
	The latter means that the probability $p$ in Algorithm~2 is also computed via QAE and enters the calculation
	of the conditional expectation. 
	In other words, the QAE-result for $p$
	occurs in the denominator to calculate the conditional expectation and by the fact that $p$ is usually close to zero, the results of RVaR and CVaR become very sensitive to inaccuracies in the calculation of $p$. This explains the large deviations in~\prettyref{fig:qpu_results}.
	
	Finally, we note that IQAE yields slightly better results than MLQAE when considering VaR and expectile, whereas MLQAE appears more robust than IQAE when applied to obtain the RVaR of the lognormal distribution.\\ 
	
	Summarizing, the results for CVaR and RVaR are affected by the noise on the quantum hardware. Instead, in the case of calculating VaR and expectiles, this noise is compensated by the bisection search algorithm and the estimated values are close to the exact values. 

	\section{Conclusion}\label{sec:conclusion}
	
	We present two quantum-based algorithms to calculate expectiles and Range Value-at-Risk. These algorithms are based on quantum amplitude estimation. A bisection search algorithm is performed on a classical computer to calculate expectiles. The objective function is evaluated by quantum amplitude estimation on a quantum device. This is in line with the idea to calculate the Value-at-Risk as described in~\textcite{woerner_quantum_2019}. On the other hand, the Range Value-at-Risk is a direct outcome of the quantum amplitude estimation. The methodology is inspired by the calculation of the Conditional Value-at-Risk in~\textcite{woerner_quantum_2019}.
	
	By two case studies, we find that the algorithms converge sufficiently fast on a simulator towards the true values, if the number of qubits to load the distribution is increased. The lowest estimation errors on a noisy simulator are obtained for the expectile by the IQAE method. The calculation of expectiles turns out to be robust against noises stemming from the noisy simulator. On the other hand, the Range Value-at-Risk is significantly affected by this noise. 
	
	This paper is a next step to calculate risk measures via quantum computing. Risk measures that are not considered in this manuscript can be of interest for future research. For example, expectiles are special cases of so-called shortfall risk measures, see e.g.,~\textcite{follmer_stochastic_2016}. Also other utility-based risk measures could be analyzed, like e.g.,~the optimized certainty equivalent~\parencite{ben-tal_expected_1986,ben-tal_old-new_2007} or the optimal expected utility risk measure~\parencite{vinel_certainty_2017,geissel_optimal_2018}.

	\printbibliography

	\begin{appendix}\label{sec:appendix}
		
		\section{Risk measures: Mathematical background}\label{sec:mathematicalBackground}
		
		We recall some mathematical background of risk measures. In particular, we state the definition of a general risk measure and give further properties of expectiles that are important for the applicability of the bisection search algorithm in~\prettyref{sec:bisectionExpectiles}. 
		
		To do so, we assume an atomless probability space $\left(\sampleSpace,\sigmaField,\symbolProbabilityMeasure\right)$. The linear space of equivalence classes of random variables on it is denoted by $\Lpspace{0}{\sampleSpace,\sigmaField,\symbolProbabilityMeasure}$, or $\symbolLpspace{0}$ for short. We always equip such a space with the $\symbolProbabilityMeasure$-almost sure order. For $p\in(0,\infty)$, the linear space of equivalence classes of $p$-integrable random variables is denoted by $\Lpspace{p}{\sampleSpace,\sigmaField,\symbolProbabilityMeasure}$, or $\symbolLpspace{p}$ for short. Further, we denote the linear space of all essentially bounded random variables by $\Linftyspace{\sampleSpace,\sigmaField,\symbolProbabilityMeasure}$, or $\symbolLinftyspace$ for short. From now on, we assume a linear subspace $\financialPositions$ of $\symbolLpspace{0}$.\\
		
		Now, we state the definition of (monetary) risk measures. Our definition is consistent with~\textcite[Definition 2.1]{cheridito_risk_2009} and~\textcite[Definition 1.1]{geissel_optimal_2018}.
		
		\begin{defi}[Monetary risk measures]\label{defi:monetary_risk_measures}
			A function $\rho:\financialPositions \to \left(-\infty,\infty\right]$ is called a monetary risk measure, if for all $X, Y\in\financialPositions$ the following properties hold:
			\begin{enumerate}[(i)]
				\item Finiteness at $0$: $\monetaryRiskMeasure{0}\in\realLine$.
				\item Monotonicity: $X\leq Y$ implies $\rho(X) \ge \rho(Y)$.
				\item Cash invariance: For all $m\in\realLine$ we have $\rho(X+ m) = \rho (X) -m$.
			\end{enumerate}
			A monetary risk measure is called convex, if for all $X$, $Y\in\financialPositions$ it holds:
			\begin{enumerate}[(i)]
				\item[(iv)] Convexity: For $\alpha\in(0,1)$ we have $\rho(\alpha X + (1-\alpha)Y) \leq \alpha \rho(X) + (1-\alpha) \rho(Y)$.
			\end{enumerate}
			A convex monetary risk measure is called coherent, if for all $X\in\financialPositions$ it holds:
			\begin{enumerate}[(i)]
				\item[(v)] Positive homogeneity: For $\alpha\geq 0$ we have $\rho(\alpha X) =\alpha \rho(X)$.
			\end{enumerate}
		\end{defi}
		
		\begin{rem}
			\prettyref{defi:monetary_risk_measures} states that risk measures are functionals which map a financial position in form of a random variable to a key figure which should describe the risk of this financial position. These functionals satisfy desirable properties from an economical point of view. We discuss these properties shortly: Monotonicity says that the capital requirement is larger for a smaller financial position. Cash invariance says that if there is a riskless ingredient in the financial position, then the capital requirement can directly be reduced by this constant payoff. Convexity describes that diversification between two financial positions reduces the capital requirement. And positive homogeneity means that a linear scaling of a financial position can also be handled by scaling the capital requirement of the unscaled financial position. 
		\end{rem}
		
		VaR, CVaR, EVaR and RVaR are monetary risk measures. We discuss the EVaR and the RVaR in detail. First, let us state the mathematical precise definition of the EVaR.
		\begin{defi}[Expectiles and EVaR]
			Let $\level\in(0,1)$, $\financialPositions\subset\symbolLpspace{1}$ and $X\in\financialPositions$. The $\level$-expectile $\expectile{\level}{X}$ of $X$ is the unique solution of
			\begin{align*}
				\level\expectation{\max\{X-\expectile{\level}{X},0\}}=(1-\level)\expectation{-\min\{X-\expectile{\level}{X},0\}}.
			\end{align*}
			The Expectile-VaR (EVaR) at level $\level$ is then defined by $\expectileVaR{\level}{X}\defgl -\expectile{\level}{X}$.
		\end{defi}
		
		\begin{rem}
			For a choice of  $\level\leq\frac{1}{2}$, the EVaR is a coherent monetary risk measure, see e.g.,~\textcite{bellini_risk_2017}. Moreover, expectiles are the only generalized quantiles that lead to coherent monetary risk measures, see~\textcite[Proposition 6]{bellini_generalized_2014}.
		\end{rem}
		
		Now, we state another interpretation of expectiles. They are special cases of the \textit{zero utility premium principle}. Therefore, we assume a loss function $l:\realLine\rightarrow\realLine$ satisfying specific properties\footnote{For instance,~\textcite[Definition 4.111]{follmer_stochastic_2016} define a loss function to be increasing and not identically constant.} and a random variable $X\in\financialPositions$ representing the profit and loss of an agent at a future time point. The zero utility premium principle states that the premium $p$ for covering $X$ should satisfy the following equation: $\expectation{l(-X-p)}=0$.
		
		Expectiles are the special case in which $l$ is given by the following relation with respect to a scalar $\level\in\left(0,1\right)$:
		\begin{align*}
			l(x)=
			\begin{cases}
				(1-\level)x & ,x>0,\\
				\level x & ,x\leq 0.
			\end{cases}
		\end{align*}
		
		Next, we prove properties of the function $h_{X,\alpha}$, which we already mentioned in~\prettyref{sec:bisectionExpectiles}. If the lower and upper bounds of a distribution lead to different signs with respect to the function $x\mapsto h_{X,\level}\left(x\right)-x$, then they are suitable as starting points for the bisection search algorithm. The upcoming result states that this situation is satisfied.
		
		Note, on a quantum computer we use circuits to describe bounded distributions, which act as approximations for unbounded distributions. Hence, it is enough to state the upcoming result only for bounded random variables. We denote the essential infimum, respectively the essential supremum, of a random variable $X$ by $\essinf X$, respectively $\esssup X$.
		
		\begin{prop}[Starting values for bisection search algorithm]\label{prop:startingValuesDifferentSigns}
			Let $\level\in\left[\frac{1}{2},1\right)$ and $X\in\symbolLinftyspace$. Then it holds that
			\begin{align*}
				h_{X,\level}\left(\essinf X\right)\geq\essinf X\quad \text{and}\quad h_{X,\level}\left(\esssup X\right)\leq\esssup X.
			\end{align*}
		\end{prop}
		
		\begin{proof}
			For the essential infimum we obtain
			\begin{align*}
				h_{X,\level}\left(\essinf X\right)-\essinf X = \frac{\level}{1-\level}\Big(\expectation{X}-\essinf X\Big)\geq 0.
			\end{align*}
			For the essential supremum it holds that
			\begin{align*}
				h_{X,\level}\left(\esssup X\right)-\esssup X = \expectation{X}-\esssup X\leq 0.
			\end{align*}
		\end{proof}
		
		
		\begin{prop}[Continuity of $h_{X,\level}$]\label{prop:continuityOf_h}
			Let $\level\in\left[\frac{1}{2},1\right)$ and $X\in\symbolLspace$. The function $h_{X,\level}$ is continuous with respect to the Euclidean norm. 
		\end{prop}
		
		\begin{proof}
			Assume an arbitrary sequence $\left\{x_n\right\}\subset\realLine$ which converges to a point $x\in\realLine$. Then the sequence $\left\{Y_n\right\}$ defined by
			\begin{align*}
				Y_n = \max\left\{\left(1+\beta\right)X-\beta x_n,X\right\}\in\symbolLspace,
			\end{align*} 
			converges pointwise to 
			\begin{align*}
				Y = \max\left\{\left(1+\beta\right)X-\beta x,X\right\}\in\symbolLspace.
			\end{align*} 
			Furthermore, we have for each $n\in\naturalNumbers$ that
			\begin{align*}
				\left|\max\left\{\left(1+\beta\right)X-\beta x_n,X\right\}\right|\leq\left|X\right|+\left|\beta\right|\max\left\{X-\min\limits_{n\in\naturalNumbers}x_n,0\right\}\in\symbolLspace.
			\end{align*} 
			Hence, by dominated convergence $Y_n\xrightarrow{\symbolLspace}Y$,
			i.e.,~$\expectation{Y_n}\rightarrow\expectation{Y}$ and we obtain that $h_{X,\level}$ is continuous.
		\end{proof}
		
		For the sake of completeness, we also state the mathematical precise definition of the RVaR, which is orientated on~\textcite[Definition 2.1]{fissler_elicitability_2021}.
		\begin{defi}[Range Value-at-Risk]
			The Range Value-at-Risk (RVaR) of a payoff  $X\in\financialPositions\subset\symbolLpspace{1}$ at levels $0<\alpha<\beta< 1$ is defined by\footnote{Here, $\valueAtRisk{\lambda}{X}$ refers to the negative upper quantile function of a random variable $X$ at level $\lambda$. For more information about quantile functions we refer to~\textcite[Appendix A.3]{follmer_stochastic_2016}.}.
			\begin{align*}
				\rangeValueAtRisk{\alpha,\beta}{X} \defgl \frac{1}{\beta - \alpha}\,\int_\alpha^\beta \valueAtRisk{u}{X} \diff u.
			\end{align*}
		\end{defi}
		
		\begin{rem}
			We obtain the following limit behavior: $\lim_{\alpha\uparrow \beta}\rangeValueAtRisk{\alpha,\beta}{X} = \valueAtRisk{\beta}{X}$,
			i.e.,~the RVaR converges to the VaR at level $\beta$ for $\alpha\rightarrow\beta$. The Range Value-at-Risk is not coherent, due to the missing convexity.
		\end{rem}

		\section{Loading probability density functions}\label{sec:distributionLoading}
		
		To load a distribution, we perform a state transformation of the first basis vector $\ket{0}$ to a vector $\ket{a}$ representing the probability density function (PDF) of a random variable. Let us assume a number $n$ of qubits and that $\left\{p_i\right\}_{i\in\left\{0,\dots,2^n-1\right\}}$ are the density values with respect to~$\ket{a}$. Then we obtain
		\begin{align*}
			\ket{a}=\sum_{i=0}^{2^n-1}\sqrt{p_i}\ket{i}.
		\end{align*}
		
		So, we aim to find an operator $\tilde{A}$ such that $\tilde{A}\ket{0}=\ket{a}$. The procedure in~\textcite{mottonen_transformation_2004} describes the calculation of an operator $A$ such that~$A\ket{a}=\ket{0}$. Hence, $\tilde{A}$ is equal to the complex conjugate of $A$, which we denote by $A^{\dagger}$.
		
		For the sake of brevity, we use the same notations as in~\textcite[Section III]{mottonen_transformation_2004} without introducing them explicitly. Their procedure is more general and two steps are performed. In our case, the first step is not necessary, because we only consider real values $p_i$, i.e.,~$\omega_1=...=\omega_N=0$.
		
		Therefore, we rely on the second step, which performs multiple uniformly controlled y-rotations. The term uniformly says that the y-rotations are uniformly distributed with respect to the possible states of the controlled qubits. We denote a uniformly controlled y-rotation for a target qubit $m$, a number $k$ of controlled qubits and a vector $\alpha$ of angles by $F_m^k(y,\alpha)$. 
		
		The operator $A$ is calculated via~\textcite[Eq.~(7)]{mottonen_transformation_2004}, which simplifies to
		\begin{align}\label{eq:operator_A}
			A\ket{a}=\left(\prod_{j=1}^{n}F_j^{j-1}(y,\alpha_{n-j+1})\otimes I_{2^{n-j}}\right)\ket{f}=\ket{0},
		\end{align}
		with angles 
		\begin{align}\label{eq:operator_A_angles}
			\alpha_{i,n-j+1}=2\arcsin\left(\sqrt{\frac{\sum\limits_{l=0}^{2^{n-j}-1}p_{(2i-1)2^{n-j}+l}}{\sum\limits_{l=0}^{2^{n-j+1}-1}p_{(i-1)2^{n-j+1}+l}}}\right),\quad i\in\left\{1,\dots,2^{j-1}\right\}.
		\end{align}
		
		In our setting, we can interpret these angles as (conditional) probabilities as the next example shows.		
		\begin{exam}[$j=2,i\in\left\{1,2\right\}$]
			Here, we have one control qubit and therefore two angles:
			\begin{align*}
				&\alpha_{1,n-1}=2\arcsin\left(\sqrt{\frac{\sum\limits_{l=0}^{2^{n-2}-1}p_{2^{n-2}+l}}{\sum\limits_{l=0}^{2^{n-1}-1}p_{l}}}\right),
				&\alpha_{2,n-1}=2\arcsin\left(\sqrt{\frac{\sum\limits_{l=0}^{2^{n-2}-1}p_{(^{n-1}+2^{n-2}+l}}{\sum\limits_{l=0}^{2^{n-1}-1}p_{2^{n-1}+l}}}\right).
			\end{align*}
			The expressions in the square roots are conditional probabilities. For instance, $\alpha_{1,n-1}$ is the probability to realize one of the values $\left\{2^{n-2},2^{n-2}+1,\dots,2^{n-1}-1\right\}$ (right half of interval $\left[0,2^{n-1}-1\right]$) under the condition that one of the values $\left\{0,1,\dots,2^{n-1}-1\right\}$ is realized.
		\end{exam}

	\begin{figure}[htp]
		\centering
		\begin{subfigure}{0.44\textwidth}
			\[
			\begin{array}{c}
				\Qcircuit @C=1.5em @R=1.2em {
					&\ctrlbo{1} & \ctrlbo{1} & \ctrlbo{1} & \gate{R_y^1} & \qw\\
					&\ctrlbo{1} & \ctrlbo{1} & \gate{R_y^2} & \qw & \qw\\
					&\ctrlbo{1} & \gate{R_y^3} & \qw        & \qw & \qw\\
					&\gate{R_y^4} & \qw        & \qw        & \qw & \qw 
				}
			\end{array}
			\]
			\caption{Illustration of operator $A$ for $n=4$ qubits. It is a sequence of uniformly controlled y-rotations. $R_y^j$ corresponds to the index $j$ in the product in~\prettyref{eq:operator_A}.}
			\label{fig:operator_A}		
		\end{subfigure}
		\hfill
		\begin{subfigure}{0.54\textwidth}
			\[
			\begin{array}{c}
				\Qcircuit @C=1.5em @R=1.2em {
					&\gate{\bar{R}_y^1} & \ctrlbo{1}         & \ctrlbo{1}         & \ctrlbo{1} & \qw\\
					&\qw                & \gate{\bar{R}_y^2} & \ctrlbo{1}         & \ctrlbo{1} & \qw\\
					&\qw                & \qw                & \gate{\bar{R}_y^3} & \ctrlbo{1} & \qw\\
					&\qw                & \qw                & \qw                & \gate{\bar{R}_y^4} & \qw 
				}
			\end{array}
			\]
			\caption{Illustration of operator $A^{\dagger}$ for $n=4$ qubits. It is a sequence of the conjugates of the uniformly controlled y-rotations, which we denote by $\bar{R}_y^j$.}
			\label{fig:operator_conjugate_A}
		\end{subfigure}
		\caption{Illustrations of operators $A$ and $A^{\dagger}$ for $n=4$ qubits.}
		\label{fig:operators_A_and_conjugate_A}
	\end{figure}
	
		Next, we develop $A^{\dagger}$. In~\prettyref{fig:operator_A} we illustrate the operator $A$. Here, we use the same symbols as in~\textcite{mottonen_transformation_2004} to illustrate uniformly controlled y-rotations, compare also~\prettyref{fig:uniformly_controlled_rotations}.
		
		As a consequence of~\prettyref{fig:operator_A} the conjugate operator $A^{\dagger}$ corresponds to the circuit in~\prettyref{fig:operator_conjugate_A}. Therefore, we rely on the conjugates of the uniformly controlled y-rotations, which are  uniformly controlled y-rotations with modified angles. So, if we understand the conjugate of a uniformly controlled y-rotation, we know the operator $A^{\dagger}$ and we can load the PDF. Therefore, in a first step we explain in more detail the construction of a circuit for a uniformly controlled y-rotation. In a second step we build its conjugate.\newline
		
		\textbf{Step 1:} A uniformly controlled y-rotation is illustrated on the left-hand side in~\prettyref{fig:uniformly_controlled_rotations}. So, we see that it consists of a rotation with different angles $\alpha_0,\dots,\alpha_ {2^{n}-1}$ for the possible states of the control qubits. The idea in~\textcite{mottonen_quantum_2004} is to reformulate the circuit using simple rotations with angles $\theta_0,\dots,\theta_{2^{n}-1}$ and CNOT-gates, see right-hand side in~\prettyref{fig:uniformly_controlled_rotations}.
		
		\begin{figure}[htp]
			\centering
			\begin{subfigure}{0.44\textwidth}
				\includegraphics[width=\linewidth]{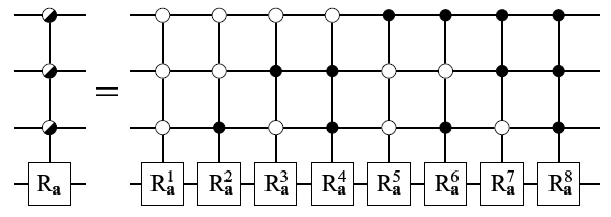}
				\caption{Uniformly controlled rotation around axis $a$ for three control qubits. Black (respectively white) control bits stand for $1$ (respectively $0$).}		
			\end{subfigure}
			\hfill
			\begin{subfigure}{0.54\textwidth}
				\includegraphics[width=\linewidth]{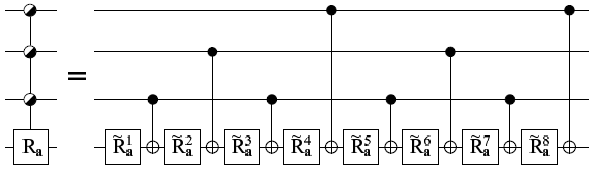}
				\caption{Uniformly controlled rotation via simple rotations and CNOT-gates.}
			\end{subfigure}
			\caption{\textit{Left-hand side:} Taken from~\textcite[Figure 1]{mottonen_quantum_2004}. \textit{Right-hand side:} Taken from~\textcite[Figure 2 (a)]{mottonen_quantum_2004}.}
			\label{fig:uniformly_controlled_rotations}
		\end{figure}
		
		To translate the $\alpha$-angles into the $\theta$-angles we have to solve a linear system of equations. In the case of $k$ control qubits we write it as $M^k(\theta_0\, \dots\, \theta_{2^n-1})^{\intercal}
		=(\alpha_0\, \dots\,\alpha_{2^n-1})^{\intercal}$.
		
		A row in this system stands for one possible state of the control qubits. Hence, the matrix $M^k$ is calculated by requiring that the rotations in both circuits in~\prettyref{fig:uniformly_controlled_rotations} are the same for all possible control qubit states. We can use binary and gray codes to calculate $M^k$. Hence, let us denote by $b_i$ the binary representation of $i$ and by $g_j$ the $j$-bit string in the binary gray code. Then, it holds that $M^k_{ij}=(-1)^{b_{i-1}\cdot g_{j-1}}.$

		Row $i$ of $M^k$ indicates the state of the control qubits given by $b_{i-1}$. $b_{i-1}$ tells us which of the control qubits are active and which are not. The columns of $M^k$ indicate the digit sequences of the gray code. 
		
		Now one could wonder why this procedure for determining $M^k$ works. We only exemplify the reason for it.   In~\prettyref{fig:gray_code} the gray code corresponding to the circuit in~\prettyref{fig:uniformly_controlled_rotations} is illustrated, i.e.,~three control qubits are used. Assume that only the first and third control qubits are active. Note, that the y-rotation for $\theta_0$ is not affected by the control qubits. But in the next time step the gray code changes for the first control qubit. This is the transition from $g_0$ to $g_1$. Since the first qubit is active, the NOT-gate is applied. So we definitely perform the y-rotation with $-\theta_1$. But when does the next NOT-gate is applied? Or in other words: When does the gray code changes the next time for an active control qubit? This is not the case for $g_2$, because the second control qubit is not active. But it is the case for $g_3$, because the third control qubit is active. Hence, we apply y-rotations with respect to $-\theta_2$ and $\theta_3$. 
		
		\begin{figure}[htp]
			\includegraphics[scale=0.8]{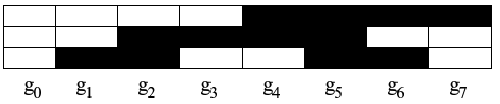}
			\caption{Gray code for three control bits. The bottom (respectively top) row corresponds to the first (respectively third) control qubit. Taken from~\textcite[Figure 2 (b)]{mottonen_quantum_2004}.}		
			\label{fig:gray_code}
		\end{figure}
		
		To determine $\theta$ from $\alpha$ we have to find the inverse of $M^k$. Now $2^{-\frac{k}{2}}M^k$ is orthogonal, see~\textcite{mottonen_quantum_2004}. Hence, $2^{-k}\left(M^k\right)^{T}M^k=\left(2^{-\frac{k}{2}}M^k\right)^{T} 2^{-\frac{k}{2}}M^k=I$, and $\left(M^k\right)^{-1}=2^{-k}\left(M^k\right)^{T}$. Finally, we are in a position to perform a uniformly controlled y-rotation by only knowing the $\alpha$-values.
		
		\textbf{Step 2:} Remember, our aim is to determine the conjugate of a uniformly controlled y-rotation. For brevity, we only exemplify it for two control qubits, i.e.,~$k=2$. The conjugate is given by the circuit in~\prettyref{fig:uniformly_controlled_rotations_2_control_qubit_conjugate}.
		
		\begin{figure}[htp]
			\[
			\begin{array}{c}
				\Qcircuit @C=1.5em @R=1.2em {
					& \ctrl{2} & \qw & \qw & \qw & \ctrl{2} & \qw & \qw & \qw & \qw \\
					& \qw & \qw & \ctrl{1} & \qw & \qw & \qw & \ctrl{1} & \qw & \qw        \\
					& \targ & \gate{R_y\left(-\theta_3\right)} & \targ & \gate{R_y\left(-\theta_2\right)} & \targ & \gate{R_y\left(-\theta_1\right)} & \targ & \gate{R_y\left(-\theta_0\right)} & \qw\\
				}
			\end{array}
			\]
			\caption{Conjugate of uniformly controlled y-rotation with $2$ control qubits.}
			\label{fig:uniformly_controlled_rotations_2_control_qubit_conjugate}
		\end{figure}
		
		As mentioned in~\textcite[page 2]{mottonen_quantum_2004}, the uniformly controlled rotation ``can also be achieved by a horizontally mirrored version of the quantum circuit presented''. The circuit in~\prettyref{fig:uniformly_controlled_rotations_2_control_qubit_conjugate} is of the same form. Hence, the conjugate can also be mirrored. Therefore, we obtain the horizontally mirrored version of the conjugate in~\prettyref{fig:uniformly_controlled_rotations_2_control_qubit_conjugate_2}. This is a uniformly controlled y-rotation with modified angles.
		
		\begin{figure}[htp]
			\[
			\begin{array}{c}
				\Qcircuit @C=1.5em @R=1.2em {
					& \qw & \qw & \qw & \ctrl{2} & \qw & \qw & \qw & \ctrl{2} & \qw \\
					& \qw & \ctrl{1} & \qw & \qw & \qw & \ctrl{1} & \qw & \qw & \qw        \\
					& \gate{R_y\left(-\theta_0\right)} & \targ & \gate{R_y\left(-\theta_1\right)} & \targ & \gate{R_y\left(-\theta_2\right)} & \targ & \gate{R_y\left(-\theta_3\right)} & \targ & \qw\\
				}
			\end{array}
			\]
			\caption{Horizontally mirrored version of the conjugate of the uniformly controlled y-rotation with $2$ control qubits.}
			\label{fig:uniformly_controlled_rotations_2_control_qubit_conjugate_2}
		\end{figure}
		
		In a last step, we calculate the new $\alpha$-values corresponding to the negative angles $-\theta$. For the negative value of a $\theta$-vector it holds that $-\theta = -2^{-k}\left(M^k\right)^{T}\alpha = 2^{-k}\left(M^k\right)^{T}\left(-\alpha\right)$. Hence, we are interested in a compact formula for $-\alpha$. We obtain with the help of~\prettyref{eq:operator_A_angles} that
		\begin{align*}
			-\alpha_{i,n-j+1}=2\arccos\left(\sqrt{\frac{\sum\limits_{l=0}^{2^{n-j}-1}p_{(2i-1)2^{n-j}+l}}{\sum\limits_{l=0}^{2^{n-j+1}-1}p_{(i-1)2^{n-j+1}+l}}}\right)-\pi.
		\end{align*} 
		
		If we transform the $\alpha$-values to the $\theta$-values we obtain a $\pi$-shift of the form $c\pi$ with an even integer $c$. Hence, this $\pi$-shift corresponds to multiple $360^{\circ}$ y-rotations and we can omit it in the calculation of $-\alpha_{i,n-j+1}$. Therefore, we finally work with the following modified angles:
		\begin{align*}
			\tilde{\alpha}_{i,n-j+1}=2\arccos\left(\sqrt{\frac{\sum\limits_{l=0}^{2^{n-j}-1}p_{(2i-1)2^{n-j}+l}}{\sum\limits_{l=0}^{2^{n-j+1}-1}p_{(i-1)2^{n-j+1}+l}}}\right).
		\end{align*}
		
		\section{Figures and tables}\label{sec:figures}
		
		The figures show the convergence behavior with respect to the lognormal distribution.
		
		\begin{figure}[!ht]
			\centering
			\includegraphics[scale=0.4]{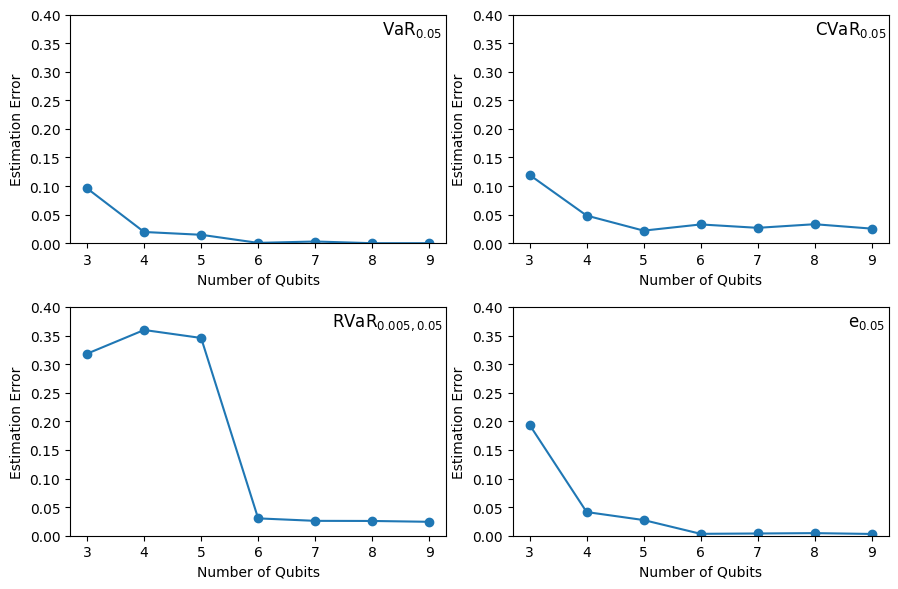}
			\caption{The estimation error on a simulator as a function of the number of qubits used 
				for loading the distribution. The error is given relatively to the length of the domain on 
				which $f$ is defined. For the results of this plot we applied the IQAE with 
				target precision $0.05$ and confidence level $0.01$ to $LN(9.6754,0.7416^2)$.  
				We used $\gamma = \pi/8$ for VaR and CVaR and $\gamma = \pi/4$ for RVaR and expectiles.}
			\label{fig:sim_result_appendix_lognorm}
		\end{figure}
		
		\begin{figure}[!ht]
			\centering
			\includegraphics[scale=0.45]{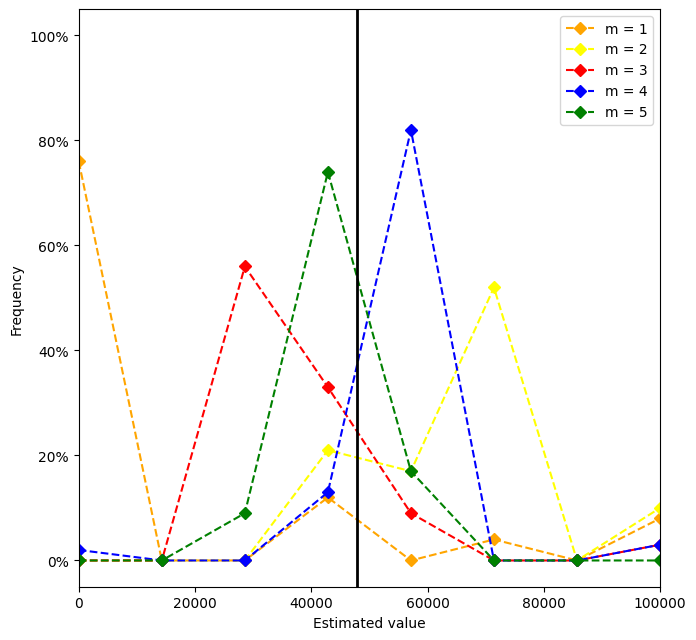}
			\caption{Estimation for the expectile calculated with the canonical QAE on a simulator.
				We used $\gamma = \pi/4$. The underlying distribution is $LN(9.6754,0.7416^2)$. 
				The exact value is indicated by the black line. With an increased number of ancilla qubits $m$, the most frequent result approaches the exact value.}
			\label{fig:canonical_lognorm}
		\end{figure}		
		
	\end{appendix}
		
\end{document}

%% file: config/packages.tex
\usepackage{caption}
\usepackage[english]{babel} 
\usepackage[utf8]{inputenc}
\usepackage[T1]{fontenc}
\usepackage{lmodern}
\usepackage{authblk}
\usepackage{setspace}

\usepackage[backend=biber, style=apa]{biblatex} 
\usepackage[autostyle]{csquotes} 
\usepackage{url}
\usepackage[colorlinks, allcolors=blue]{hyperref}
\hypersetup{breaklinks=true}
\AtEveryBibitem{
	\clearfield{note}
}

\usepackage{abstract}
\usepackage{enumerate}
\usepackage{appendix}
\usepackage{mathtools}
\usepackage{amsmath}
\usepackage{amsthm}   
\usepackage{amsfonts}
\usepackage{amssymb}
\usepackage{textcomp} 
\usepackage{wasysym}
\usepackage{color}
\usepackage[dvipsnames]{xcolor}
\usepackage{tikz}
\usetikzlibrary{fit, patterns, positioning, shapes, calc, arrows, backgrounds}

\usepackage{physics}
\usepackage{qcircuit}

\usepackage{algorithm}
\usepackage{algpseudocode}

\tikzset{
	assets/.style={
		rectangle split,
		rectangle split parts=1,
		rectangle split part fill={blue!20},
		rounded corners,
		draw=black, very thick,
		minimum height=10cm,
		minimum width=5cm,
		text width=2cm,
		text centered,
	}
}
\tikzset{
	state/.style={
		rectangle split,
		rectangle split parts=2,
		rectangle split part fill={red!30,blue!20},
		rounded corners,
		draw=black, very thick,
		minimum height=2em,
		text width=3cm,
		inner sep=2pt,
		text centered,
	}
}

\usepackage{pgfplots}
\pgfplotsset{%
	every tick label/.append style = {font=\tiny},
	every axis label/.append style = {font=\scriptsize}
}
\usepackage{graphicx}  
\usepackage{subcaption}
\usepackage{changes}
\usepackage{hhline}
\usepackage{makecell}

\usepackage{placeins}
\usepackage[most]{tcolorbox}

\makeatletter
\renewenvironment{proof}[1][\proofname]{\par
	\pushQED{\qed}%
	\normalfont \topsep6\p@\@plus6\p@\relax
	\trivlist
	\item[\hskip\labelsep
	\itshape
	#1\@addpunct{:}]\ignorespaces
}{%
	\popQED\endtrivlist\@endpefalse
}
\makeatother

\usepackage{prettyref}
\usepackage{titleref}
\newrefformat{cha}{Chapter~\ref{#1}}
\newrefformat{sec}{Section~\ref{#1}}
\newrefformat{app}{Appendix~\ref{#1}}
\newrefformat{fig}{Figure~\ref{#1}}
\newrefformat{tab}{Table~\ref{#1}}
\newrefformat{defi}{Definition~\ref{#1}}
\newrefformat{assump}{Assumption~\ref{#1}}
\newrefformat{theorem}{Theorem~\ref{#1}}
\newrefformat{prop}{Proposition~\ref{#1}}
\newrefformat{lem}{Lemma~\ref{#1}}
\newrefformat{cor}{Corollary~\ref{#1}}
\newrefformat{rem}{Remark~\ref{#1}}
\newrefformat{exam}{Example~\ref{#1}}
\newrefformat{eq}{(\ref{#1})}

\newtheoremstyle{italic}
{15pt}
{15pt}
{\normalfont \itshape}
{}
{\normalfont \bfseries}
{.}
{ }
{}

\newtheoremstyle{nonItalic}
{15pt}
{15pt}
{}
{}
{\normalfont \bfseries}
{.}
{ }
{}

\theoremstyle{italic} 
\newtheorem{defi}{Definition}[section]

\newtheorem{prop}[defi]{Proposition}

\theoremstyle{nonItalic}
\newtheorem{rem}[defi]{Remark}
\newtheorem{exam}[defi]{Example}

%% file: config/commands.tex
\newcommand{\defgl}{\mathrel{\mathop:\!\!=}}

\newcommand{\naturalNumbers}{\mathbb{N}}
\newcommand{\realLine}{\mathbb{R}}
\newcommand*\diff{\mathop{}\!\mathrm{d}}

\newcommand{\sampleSpace}{\Omega}
\newcommand{\sigmaField}{\mathcal{F}}
\newcommand{\symbolProbabilityMeasure}{P}
\newcommand{\expectation}[1]{E[#1]}

\newcommand{\symbolLpspace}[1]{L^{#1}}
\newcommand{\Lpspace}[2]{\symbolLpspace{#1}(#2)}
\newcommand{\symbolLspace}{L^{1}}

\newcommand{\symbolLinftyspace}{L^{\infty}}
\newcommand{\Linftyspace}[1]{\symbolLinftyspace(#1)}

\newcommand{\financialPositions}{\mathcal{X}}
\newcommand{\symbolMonetaryRiskMeasure}{\rho}
\newcommand{\monetaryRiskMeasure}[1]{\symbolMonetaryRiskMeasure(#1)}

\DeclareMathOperator{\symbolOperatorValueAtRisk}{VaR}

\DeclareMathOperator{\symbolOperatorConditionalValueAtRisk}{CVaR}
\DeclareMathOperator{\symbolOperatorRangeValueAtRisk}{RVaR}
\DeclareMathOperator{\symbolOperatorExpectileValueAtRisk}{EVaR}
\DeclareMathOperator{\symbolOperatorExpectile}{e}
\newcommand{\symbolValueAtRisk}[1]{\symbolOperatorValueAtRisk_{#1}}

\newcommand{\symbolConditionalValueAtRisk}[1]{\symbolOperatorConditionalValueAtRisk_{#1}}
\newcommand{\symbolRangeValueAtRisk}[1]{\symbolOperatorRangeValueAtRisk_{#1}}
\newcommand{\valueAtRisk}[2]{\symbolValueAtRisk{#1}(#2)}

\newcommand{\conditionalValueAtRisk}[2]{\symbolConditionalValueAtRisk{#1}(#2)}
\newcommand{\rangeValueAtRisk}[2]{\symbolRangeValueAtRisk{#1}(#2)}

\DeclareMathOperator*{\essinf}{ess\,inf}
\DeclareMathOperator*{\esssup}{ess\,sup}

\newcommand{\level}{\alpha}
\newcommand{\expectile}[2]{\symbolOperatorExpectile_{#1}(#2)}
\newcommand{\expectileVaR}[2]{\symbolOperatorExpectileValueAtRisk_{#1}(#2)}

\newcommand*\halfcirc[1][1ex]{%
	\begin{tikzpicture}
		\draw[fill] (0,0)-- (225:#1) arc (225:405:#1) -- cycle ;
		\draw (0,0) circle (#1);
\end{tikzpicture}}

\newcommand{\controlbo}{*+<.01em>{\halfcirc[0.6ex]}}
\newcommand{\ctrlbo}[1]{\controlbo \qwx[#1] \qw}